\documentclass[a4paper,11pt]{article}
\pdfoutput=1 
\usepackage{jheppub}
\usepackage{etoolbox}
\usepackage[T1]{fontenc}

\usepackage{subfig}

\usepackage{bm}

\newcommand{\area}{\text{area}}

\usepackage{tikz}
\usetikzlibrary{decorations.markings}
\usetikzlibrary{positioning, calc}
\usetikzlibrary{calc}
\usetikzlibrary{arrows}
\usepackage{tikz-3dplot}
\usepackage{graphicx}
\usepackage{mdwlist}
\usepackage{color}

\newcommand{\A}{\mathcal{A}}
\newcommand{\B}{\mathcal{B}}
\newcommand{\C}{\mathcal{C}}
\newcommand{\R}{\mathcal{R}}
\newcommand{\V}{\mathcal{V}}
\newcommand{\U}{\mathcal{U}}

\newtheorem{theorem}{Theorem}

\newtheorem{definition}[theorem]{Definition}

\newtheorem{lemma}[theorem]{Lemma}

\newtheorem{protocol}[theorem]{Protocol}

\newenvironment{proof}[1][Proof]{\noindent\textbf{#1.} }{\ \rule{0.5em}{0.5em}}
\newenvironment{argument}[1][Argument]{\noindent\textbf{#1.} }{\ \rule{0.5em}{0.5em}}

\usetikzlibrary{decorations.pathreplacing,decorations.markings}
\usepackage{mathrsfs}

\tikzset{
    >=stealth',
    punkt/.style={
           rectangle,
           rounded corners,
           draw=black, very thick,
           text width=6.5em,
           minimum height=2em,
           text centered},
    pil/.style={
           ->,
           thick,
           shorten <=2pt,
           shorten >=2pt,},
  on each segment/.style={
    decorate,
    decoration={
      show path construction,
      moveto code={},
      lineto code={
        \path [#1]
        (\tikzinputsegmentfirst) -- (\tikzinputsegmentlast);
      },
      curveto code={
        \path [#1] (\tikzinputsegmentfirst)
        .. controls
        (\tikzinputsegmentsupporta) and (\tikzinputsegmentsupportb)
        ..
        (\tikzinputsegmentlast);
      },
      closepath code={
        \path [#1]
        (\tikzinputsegmentfirst) -- (\tikzinputsegmentlast);
      },
    },
  },
  mid arrow/.style={postaction={decorate,decoration={
        markings,
        mark=at position .5 with {\arrow[#1]{stealth'}}
      }}}
}

 \newcommand{\ket}[1]{|#1\rangle}
\newcommand{\ketbra}[2]{|#1\rangle\!\langle#2|}

\usetikzlibrary{fadings}

\title{Bulk private curves require large conditional mutual information}
\author[a]{Alex May}
\affiliation[a]{The University of British Columbia}
\emailAdd{may@phas.ubc.ca}

\abstract{We prove a theorem showing that the existence of ``private'' curves in the bulk of AdS implies two regions of the dual CFT share strong correlations. A private curve is a causal curve which avoids the entanglement wedge of a specified boundary region $\mathcal{U}$. The implied correlation is measured by the conditional mutual information $I(\mathcal{V}_1:\mathcal{V}_2\,|\,\mathcal{U})$, which is $O(1/G_N)$ when a private causal curve exists. The regions $\mathcal{V}_1$ and $\mathcal{V}_2$ are specified by the endpoints of the causal curve and the placement of the region $\mathcal{U}$. This gives a causal perspective on the conditional mutual information in AdS/CFT, analogous to the causal perspective on the mutual information given by earlier work on the connected wedge theorem. We give an information theoretic argument for our theorem, along with a bulk geometric proof. In the geometric perspective, the theorem follows from the maximin formula and entanglement wedge nesting. In the information theoretic approach, the theorem follows from resource requirements for sending private messages over a public quantum channel.}

\begin{document} 

\maketitle

\section{Introduction}

In the context of the AdS/CFT correspondence, the Ryu-Takayanagi formula \cite{ryu2006holographic} and its covariant generalizations \cite{hubeny2007covariant,wall2014maximin,engelhardt2015quantum,akers2020quantum} have deepened our understanding of how geometry and gravitational physics can be recorded into quantum mechanical degrees of freedom. Other proposals relate a variety of bulk geometric quantities to boundary entanglement, see \cite{umemoto2018entanglement,nguyen2018entanglement,dutta2019canonical,miyaji2015distance,dong2021holographic} for an incomplete listing. In all cases, boundary correlation is related to bulk spacelike surfaces. Recently, a qualitatively different connection between boundary entanglement and bulk geometry has been found. This connection is captured, at least in part, by the connected wedge theorem \cite{may2019quantum,may2020holographic,may2021quantum,may2021holographic}, which relates bulk \emph{light cones} rather than spacelike surfaces to boundary correlation. 

In this article we expand on this lightcone-entanglement connection by introducing a second theorem relating the geometry of light cones to boundary entanglement. We illustrate the new theorem in figure \ref{fig:privacyduality}. The theorem relates the existence of a special class of causal curves in the bulk to the quantum conditional mutual information of associated boundary regions becoming $O(1/G_N)$.

\begin{figure}
    \centering
    \subfloat[\label{fig:boundaryregions}]{
    \begin{tikzpicture}[scale=1.6]

    \draw (-2,0) -- (2,0) -- (2,2) -- (-2,2) -- (-2,0);
    
    \draw[fill=red!60!,opacity=0.8] (0,0.5) -- (0.5,1) -- (0,1.5) -- (-0.5,1) -- (0,0.5);
    \draw[thick] (0,0.5) -- (0.5,1) -- (0,1.5) -- (-0.5,1) -- (0,0.5);
    \node at (0,1) {$\U$};
    
    \draw[fill=red!60!,opacity=0.8] (-2,0.5) -- (-1.5,1) -- (-2,1.5) -- (-2,0.5);
    \draw[thick] (-2,0.5) -- (-1.5,1) -- (-2,1.5) -- (-2,0.5);
    \draw[fill=red!60!,opacity=0.8] (2,0.5) -- (1.5,1) -- (2,1.5) -- (2,0.5);
    \draw[thick] (2,0.5) -- (1.5,1) -- (2,1.5) -- (2,0.5);
    
    \node at (1.82,1) {$\U$};
    \node at (-1.82,1) {$\U$};
    
    \draw[fill=blue!60!,opacity=0.8] (1,0.5) -- (1.5,1) -- (1,1.5) -- (0.5,1) -- (1,0.5);
    \draw[thick] (1,0.5) -- (1.5,1) -- (1,1.5) -- (0.5,1) -- (1,0.5);
    \node at (1,1) {$\V_1$};
    
    \draw[fill=blue!60!,opacity=0.8] (-1,0.5) -- (-1.5,1) -- (-1,1.5) -- (-0.5,1) -- (-1,0.5);
    \draw[thick] (-1,0.5) -- (-1.5,1) -- (-1,1.5) -- (-0.5,1) -- (-1,0.5);
    \node at (-1,1) {$\V_2$};
    
    \node at (0,-1) {$ $};
    
    \draw[fill=gray,opacity=0.5, thick, black] (1,0) -- (1.2,0.2) -- (1,0.4) -- (0.8,0.2) -- (1,0);
    \node[below] at (1,0) {$\C$};
    
    \draw[fill=gray,opacity=0.5, thick, black] (-1,2) -- (-1.2,1.8) -- (-1,1.6) -- (-0.8,1.8) -- (-1,2);
    \node[above] at (-1,02) {$\R$};
    
    \end{tikzpicture}
    }
    \hfill
    \centering
    \subfloat[\label{fig:privacytheorem}]{
    \tdplotsetmaincoords{15}{0}
    \begin{tikzpicture}[scale=1.4,tdplot_main_coords]
    \tdplotsetrotatedcoords{0}{45}{0}
    \draw[gray] (-2,0,0) -- (-2,4,0);
    \draw[gray] (2,0,0) -- (2,4,0);
    
    \begin{scope}[tdplot_rotated_coords]
    
    \draw[domain=0:45,variable=\x,smooth, fill=red!60!,opacity=0.8] plot ({-2*sin(\x)}, {1+\x/45}, {2*cos(\x)}) -- plot ({-2*sin((45-\x))}, {3-(45-\x)/45}, {2*cos(45-\x)}) --  plot ({2*sin(\x)}, {3-\x/45}, {2*cos(\x)}) -- plot ({2*sin(45-\x)}, {1+(45-\x)/45}, {2*cos(45-\x)});
    
    \begin{scope}[canvas is xz plane at y=0]
    \draw[gray] (0,0) circle[radius=2] ;
    \end{scope}
    
    \begin{scope}[canvas is xz plane at y=4]
    \draw[gray] (0,0) circle[radius=2] ;
    \end{scope}
    
    \begin{scope}[canvas is xz plane at y=2]
    \draw[gray] (0,0) circle (2);
    \end{scope}
    
    \draw [domain=-45:45] plot ({2*cos(\x+90)},2, {2*sin(\x+90)});
    
    \draw[domain=0:45,variable=\x,smooth, fill=red!50!,opacity=0.8] plot ({-2*sin(\x+180)}, {1+\x/45}, {2*cos(\x+180)}) -- plot ({-2*sin((45-\x)+180)}, {3-(45-\x)/45}, {2*cos(45-\x+180)}) --  plot ({2*sin(\x+180)}, {3-\x/45}, {2*cos(\x+180)}) -- plot ({2*sin(45-\x+180)}, {1+(45-\x)/45}, {2*cos(45-\x+180)});
    
    \begin{scope}[canvas is xz plane at y=2]
    \draw [domain=-180:180] plot ({2*cos(\x-90)}, {2*sin(\x-90)});
    \draw [domain=-45:-135,blue,ultra thick] plot ({2*cos(\x-90)}, {2*sin(\x-90)});
    \end{scope}
    
    \draw[domain=45:135, thick] plot
    ({2*cos(\x)},{2},{2*sin(\x)-2.85});
    \draw[domain=45:135, thick] plot
    ({2*cos(\x)},{2},{-2*sin(\x)+2.85});
    
    \foreach \x in {45,...,135}
    {
    \draw[red,opacity=0.3] (0,3,2) -- ({2*cos(\x)},{2},{-2*sin(\x)+2.85});
    \draw[red,opacity=0.3] (0,3,-2) -- ({2*cos(\x)},{2},{2*sin(\x)-2.85});
    }
    
    \foreach \x in {45,...,135}
    {
    \draw[red,opacity=0.3] (0,1,2) -- ({2*cos(\x)},{2},{-2*sin(\x)+2.85});
    \draw[red,opacity=0.3] (0,1,-2) -- ({2*cos(\x)},{2},{2*sin(\x)-2.85});
    }
    
    \foreach \x in {11,14,...,41}
    {
    \draw[red,opacity=0.3] plot (0,1,-2) -- ({2*sin(\x+180)}, {1+(\x)/45}, {2*cos(\x+180)});
    }

    \foreach \x in {11,14,...,41}
    {
    \draw[red,opacity=0.3] plot (0,1,2) -- ({2*sin(\x+180)}, {1+(\x)/45}, {-2*cos(\x+180)});
    }
    
    \foreach \x in {11,14,...,41}
    {
    \draw[red,opacity=0.3] plot (0,1,2) -- ({-2*sin(\x+180)}, {1+(\x)/45}, {-2*cos(\x+180)});
    }
    
    \foreach \x in {11,14,...,41}
    {
    \draw[red,opacity=0.3] plot (0,3,2) -- ({-2*sin(\x+180)}, {3-(\x)/45}, {-2*cos(\x+180)});
    }
    
    \foreach \x in {11,14,...,41}
    {
    \draw[red,opacity=0.3] plot (0,3,-2) -- ({2*sin(\x+180)}, {3-(\x)/45}, {2*cos(\x+180)});
    }
    
    \foreach \x in {11,14,...,41}
    {
    \draw[red,opacity=0.3] plot (0,3,-2) -- ({-2*sin(\x+180)}, {3-(\x)/45}, {2*cos(\x+180)});
    }
    
    \node at (-2,1.4,0) {$\U$};
    \node at (-2,1.5,4) {$\U$};
    
    \draw[domain=0:45,variable=\x,smooth,thick] plot ({-2*sin(\x)}, {1+\x/45}, {2*cos(\x)});
    \draw[domain=0:45,variable=\x,smooth,thick] plot ({2*sin(\x)}, {1+\x/45}, {2*cos(\x)});
    \draw[domain=0:45,variable=\x,smooth,thick] plot ({-2*sin(\x)}, {3-\x/45}, {2*cos(\x)});
    \draw[domain=0:45,variable=\x,smooth,thick] plot ({2*sin(\x)}, {3-\x/45}, {2*cos(\x)});
    
    \draw[domain=0:45,variable=\x,smooth,thick] plot ({-2*sin(\x+180)}, {1+\x/45}, {2*cos(\x+180)});
    \draw[domain=0:45,variable=\x,smooth,thick] plot ({2*sin(\x+180)}, {1+\x/45}, {2*cos(\x+180)});
    \draw[domain=0:45,variable=\x,smooth,thick] plot ({-2*sin(\x+180)}, {3-\x/45}, {2*cos(\x+180)});
    \draw[domain=0:45,variable=\x,smooth,thick] plot ({2*sin(\x+180)}, {3-\x/45}, {2*cos(\x+180)});
    
    \foreach \x in {100,...,250}
    {
    \draw[gray,opacity=0.3] (2,0,0) -- ({0.55*cos(\x)+2.1},{0.33},{-0.55*sin(\x)});
    \draw[gray,opacity=0.3] (2,0.66,0) -- ({0.55*cos(\x)+2.1},{0.33},{-0.55*sin(\x)});
    }
    
    \draw[black] plot [mark=*, mark size=1] coordinates{(1.62,0.35,-0.1)};
    \node[below] at (2,0,0) {$\C$};
    \node[left] at (1.62,0.35,-0.2) {$c$};
    
    \draw[domain=-15:15,variable=\x,smooth,thick] plot ({2*sin(\x+90)}, {0.333}, {2*cos(\x+90)});
    \draw[domain=0:15,variable=\x,smooth, thick] plot ({2*sin(\x+90)}, {\x/45}, {2*cos(\x+90)});
    \draw[domain=0:-15,variable=\x,smooth, thick] plot ({2*sin(\x+90)}, {-\x/45}, {2*cos(\x+90)});
    \draw[domain=0:-15,variable=\x,smooth, thick] plot ({2*sin(\x+90)}, {0.666+\x/45}, {2*cos(\x+90)});
    \draw[domain=0:15,variable=\x,smooth, thick] plot ({2*sin(\x+90)}, {0.666-\x/45}, {2*cos(\x+90)});
    
    \draw[domain=100:250, thick] plot ({0.55*cos(\x)+2.1},{0.33},{-0.55*sin(\x)});
    
    \foreach \x in {-72,...,72}
    {
    \draw[gray,opacity=0.3] (-2,4-0.66,0) -- ({0.55*cos(\x)-2.1},{4-0.33},{-0.55*sin(\x)});
    \draw[gray,opacity=0.3] (-2,4,0) -- ({0.55*cos(\x)-2.1},{4-0.33},{-0.55*sin(\x)});
    }
    
    \draw[domain=165:195,variable=\x,smooth,thick] plot ({2*sin(\x+90)}, {4-0.333}, {2*cos(\x+90)});
    \draw[domain=0:15,variable=\x,smooth, thick] plot ({2*sin(\x-90)}, {4-\x/45}, {2*cos(\x-90)});
    \draw[domain=0:-15,variable=\x,smooth, thick] plot ({2*sin(\x-90)}, {4+\x/45}, {2*cos(\x-90)});
    \draw[domain=0:-15,variable=\x,smooth, thick] plot ({2*sin(\x-90)}, {4-0.66-\x/45}, {2*cos(\x-90)});
    \draw[domain=0:15,variable=\x,smooth, thick] plot ({2*sin(\x-90)}, {4-0.66+\x/45}, {2*cos(\x-90)});
    
    \draw[domain=-72:72, thick] plot ({0.55*cos(\x)-2.1},{4-0.33},{-0.55*sin(\x)});
    
    \node[below] at (1,1.7,0) {$\V_1$};
    \node[above] at (-1,2.3,0) {$\V_2$};
    
    \draw[postaction={on each segment={mid arrow}},purple] (1.62,0.35,-0.1) -- (-2+0.33,4-0.35,0);
    
    \begin{scope}[canvas is xz plane at y=2]
    \draw [domain=45:135,blue,ultra thick] plot ({2*cos(\x-90)}, {2*sin(\x-90)});
    \end{scope}
    
    \draw[black] plot [mark=*, mark size=1] coordinates{(-2+0.33,4-0.37,0)};
    \node[below right] at (-2+0.33,4-0.35,0.1) {$r$};
    \node[above] at (-2,4,0) {$\R$};
    
    \end{scope}
    \end{tikzpicture}
    }
    \caption{(a) An example arrangement of regions considered in the privacy-duality theorem. We've taken a case where regions $\U$ consists of two connected components. $\V_1$ and $\V_2$ are defined as in theorem \ref{thm:main}. Notice that all causal curves from $c$ to $r$ pass through $\U$. (b) The same arrangement of regions shown in the bulk. There is a causal curve $\Gamma_P$ which passes from $c\in E_\C$ to $r\in E_\R$ without ever passing through the entanglement wedge of $\U$. This allows a secret message to travel from $\C$ to $\R$ without an eavesdropper, who has access to $\U$, being able to read it. The privacy-duality theorem gives that for the message to travel from $\C$ to $\R$ securely in the boundary picture, we must have $I(\V_1:\V_2\,|\,\U)=O(1/G_N)$. }
    \label{fig:privacyduality}
\end{figure}
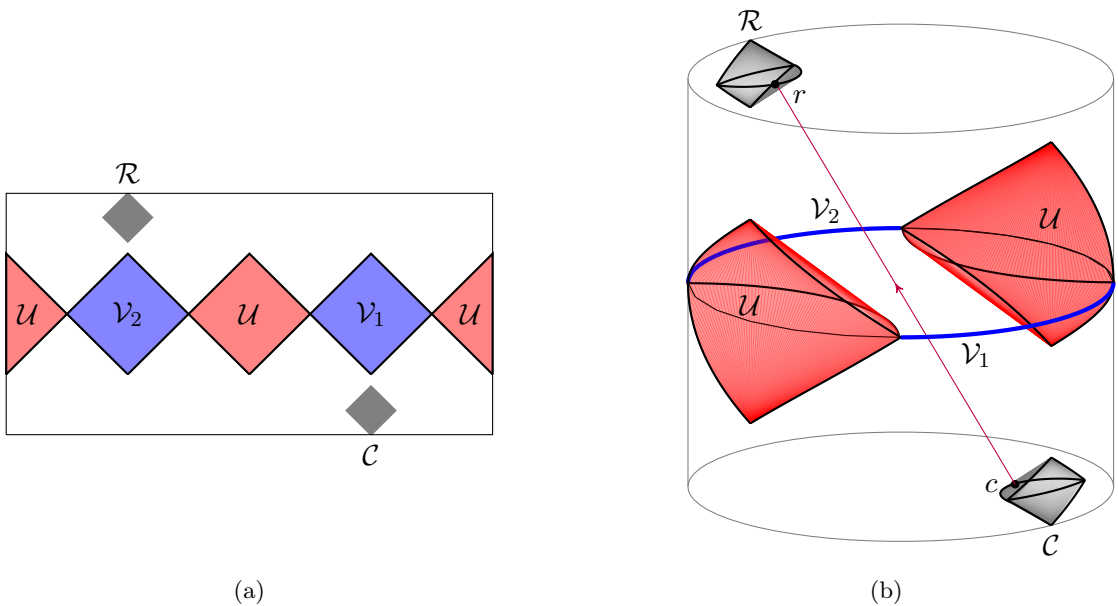

The connected wedge theorem was first found by studying quantum tasks \cite{kent2012quantum}, which are quantum computations with inputs and outputs that occur at designated spacetime locations. In \cite{may2019quantum,may2021quantum}, a general framework was proposed for applying quantum tasks to holography. The theorem introduced here also follows from the holographic quantum tasks framework, but uses a different choice of task. In particular the argument uses a task involving an eavesdropper and the preservation of privacy across the bulk and boundary descriptions. For this reason we name the new theorem the \emph{privacy-duality theorem}.

To arrive at the the privacy-duality theorem, first recall that given a boundary region $\A$ the \emph{entanglement wedge} of $\A$, labelled $E_\A$, is the bulk region which holds the same information as $\A$ \cite{czech2012gravity,headrick2014causality,wall2014maximin,jafferis2016relative,dong2016reconstruction,cotler2019entanglement}. To find $E_\A$, one looks for extremal surfaces anchored to $\A$ at the boundary and chooses the minimal area one\footnote{A more careful discussion of the entanglement wedge can be found in section \ref{sec:geometricprelim}.}. Then consider the following scenario. Three regions $\C$, $\R$, and $\U$, are defined in the boundary of an asymptotically AdS spacetime. We arrange for a message, consisting of a quantum system $X$, to be sent from a point $c$ in the entanglement wedge of $\C$ to a point $r$ in the entanglement wedge of $\R$. The message should be sent subject to the following constraint: An eavesdropper, call her Eve, may choose to access region $\U$, by for example making arbitrary measurements of $\U$. In the bulk picture this corresponds to Eve gaining access to the entanglement wedge $E_{\U}$. Our constraint is that Eve should learn nothing about the message, in which case we say the message is secure.

In some cases there is an apparent discrepancy in how hard it is to send this secret message from $\C$ to $\R$ in the bulk and boundary descriptions. For example, consider the arrangement of regions shown in figure \ref{fig:privacyduality}. In the bulk description there is a causal curve from $\C$ to $\R$ that avoids $E_{\U}$, so we can send a secure message from inside of $\C$ to $\R$ by simply moving the message along this curve. A curve which is causal and avoids $E_\U$ we will say is \emph{private with respect to $E_\U$}. Since in the bulk description the message is inaccessible to $E_{\U}$, it should be inaccessible to $\U$ in the boundary description. However, in the boundary all causal curves from $\C$ to $\R$ pass through $\U$. How then can the message travel securely from $\C$ to $\R$?

In quantum cryptography, the problem of sending secret messages has been studied in detail \cite{schumacher2006quantum,brandao2012quantum}. In their setting, the analogue of the region $\U$ is a \emph{public channel}. In particular, a quantum system $\alpha$ can be sent from $\C$ to $\R$, but there is the possibility that $\alpha$ will be received by Eve rather than the intended recipient. The general setting considered by quantum cryptographers is shown as figure \ref{fig:secrecysetup}. A resource system $\rho_{AB}$ may also be used to assist in sending the secret message. System $A$ is input, along with the message system $X$, to an encoding channel ${V}_{XA\rightarrow a\alpha}$, which we take to be an isometry. Only the system $\alpha$ is sent through the public channel while system $a$ is retained. Then, the receiver attempts to recover the message from the $\alpha B$ system. An encoding and decoding scheme for sending secret messages, along with choice of resource state $\rho_{AB}$, is known as a \emph{quantum one-time pad}. 

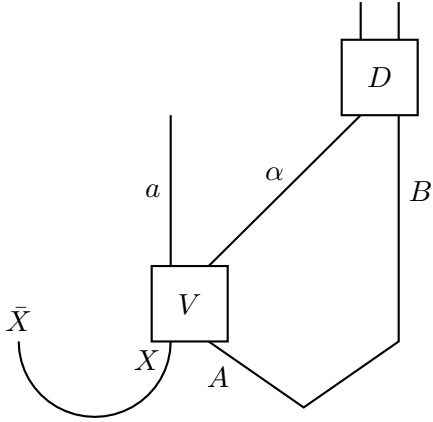
\begin{figure}
    \centering
    \begin{tikzpicture}[scale=0.5]
    
    \draw[black,thick] (0,0) rectangle (2,2);
    \node at (1,1) {$V$};
    
    \draw[black,thick] (5,6) rectangle (7,8);
    \node at (6,7) {$D$};
    
    \draw[domain=-180:0,thick] plot ({2*cos(\x)-1.5},{2*sin(\x)});
    
    \draw[black,thick] (1.5,0) -- (4,-1.75) -- (6.5,0) -- (6.5,6);
    
    \draw[black, thick] (0.5,2) -- (0.5,6); 
    \node[left] at (0.5,4) {$a$};
    
    \draw[black,thick] (5.5,8) -- (5.5,9);
    \draw[black,thick] (6.5,8) -- (6.5,9);
    
    \draw[thick] (1.5,2) -- (5.5,6);
    \node[above left] at (3.75,4) {$\alpha$};
    
    \node[right] at (6.5,4) {$B$};
    
    \node[above] at (-3.5,0) {$\bar{X}$};
    \node[left] at (0.5,-0.5) {$X$};
    
    \node[below] at (1.75,-0.4) {$A$};
    
    \end{tikzpicture}
    \caption{Circuit diagram of a general encoding and decoding procedure for sending secret messages over a public channel. The message $X$ is maximally entangled with a reference system $\bar{X}$. The encoding map $\mathcal{V}_{XA\rightarrow a\alpha}$ is applied to the message along with part of the resource state $A$, its output is divided into $a$ and $\alpha$ subsystems. $a$ is kept secure. $\alpha$ is sent over the public channel. The decoding procedure attempts to recover the $X$ system from $\alpha B$. System $\alpha$ should reveal nothing about $X$.}
    \label{fig:secrecysetup}
\end{figure}

One can prove that for the message to be sent correctly and securely, a particular pattern of correlation must be present among the $a,\alpha$ and $B$ systems. In particular the conditional mutual information $I(a:B|\alpha)$, defined by
\begin{align}
    I(a:B|\alpha) = I(B:a\alpha) - I(B:\alpha)
\end{align}
must be twice as large as the entropy of the message. The intuition for this is as follows. Take $X$ to be in a maximally entangled state with a reference system $\bar{X}$. Using this along with $V_{XA\rightarrow a\alpha}$ being an isometry, the above is equal to $I(B:A)-I(B:\alpha)$. Thus the requirement is that the sender and receiver should share correlation in their ``resource'' state on the $AB$ systems, and additionally that the $\alpha$ system sent through the public channel should not reveal too much of that correlation. 

Building on intuition stemming from the quantum cryptographers' results, we will arrive at a theorem relating bulk private curves to boundary conditional mutual information. To state it, we need to fix some notation. Given a spacetime region $\A$, we denote its causal future and past by $J^{\pm}(\A)$ when considering the bulk geometry, and $\hat{J}(\A)$ when restricting to the boundary geometry. Additionally, given a region $\A$ we will denote its boundary domain of dependence by $\hat{D}(\A)$, and its spacelike complement by $\A'$.\footnote{The boundary domain of dependence $\hat{D}(\A)$ is the set of all boundary points $p$ such that every causal curve through $p$ intersects $\A$. The spacelike complement is the set of all points spacelike separated from every point in $\A$.} Finally, we will say an extremal surface is unique if all other extremal surfaces homologous to the same boundary region have generalized entropy larger than the minimal one at order $O(1/G_N)$\footnote{Generalized entropy and homologous are defined in section \ref{sec:geometricprelim}.}. We can now state the theorem.\footnote{As discussed in more detail in section \ref{sec:geometric}, the theorem applies to holographic systems where the von Neumann entropy is calculated by the maximin formula, and the quantum focusing conjecture holds.}

\begin{theorem}\label{thm:main}
\textbf{(Privacy-duality)} Consider domains of dependence $\C$, ${\R}$, and ${\U}$ in the boundary of an asymptotically locally AdS spacetime, along with the corresponding bulk entanglement wedges $E_\C$, $E_\R$, and $E_\U$. Assume there is no private curve from $\C$ to $\R$ with respect to $\U$ in the boundary geometry. Define boundary regions
\begin{align}
    \V_1 &= \hat{D}(\hat{J}^+(\C) \cap \U' \cap \partial \Sigma), \\
    \V_2 &= \hat{D}(\hat{J}^-(\R) \cap \U' \cap \partial \Sigma),
\end{align}
where $\partial \Sigma$ is any Cauchy surface for the boundary which includes a Cauchy surface for $\U$.\footnote{Said another way, this means there exists a spacelike region $u$ such that $u\subseteq \partial \Sigma$ and $D(u)=\U$.} Then if there is a private curve in the bulk from $E_\C$ to $E_\R$ with respect to $E_\U$, and the minimal extremal surface homologous to $\U$ is unique, it follows that $I(\V_1:\V_2\,|\,\U) = O(1/G_N)$.
\end{theorem}

In the main text we give a detailed quantum information argument for this theorem. In that context, the conditional mutual information becomes $O(1/G_N)$ when there is a private curve because $O(1/G_N)$ qubits can be sent securely through the bulk geometry, which in the boundary requires condition mutual information of the same order. We can also understand the theorem geometrically. Using the HRT formula we can relate the conditional mutual information to properties of quantum extremal surfaces in the bulk. The HRT formula gives the entropy as an area term, which is $O(1/G_N)$, and a bulk entropy term. The theorem then amounts to the statement that a bulk private curve implies a non-zero area term. The proof of this assumes entanglement wedge nesting, which states that for boundary regions $\A,\B$ with $\A \subseteq \B$ we have $E_\A\subseteq E_\B$. For quantum extremal surfaces entanglement wedge nesting follows from the quantum focusing conjecture \cite{akers2020quantum}.

\vspace{0.3cm}
\noindent \textbf{Summary of notation}
\vspace{0.3cm}

\noindent We briefly summarize the notation used in this paper. 
\begin{itemize*}
    \item We denote boundary spacetime regions by script capital letters $\A,\B,...$. 
    \item The entanglement wedge of a boundary region $\A$ is denoted by $E_{\A}$.
    \item The HRT surface associated with a boundary region $\A$ we denote by $\gamma_{\A}$. We'll use $\partial$ to denote the spacelike boundary of a region, so that $\partial E_{\A} = \gamma_{\A}$
    \item We use $J^\pm(\cdot)$ for the causal future or past of a spacetime region taken in the bulk geometry, and $\hat{J}^\pm(\cdot)$ for the casual future or past taken in the boundary geometry.
    \item $D(\cdot)$ denotes the domain of dependence, taken in the bulk geometry. $\hat{D}(\cdot)$ denotes the domain of dependence taken in the boundary geometry. 
    \item We use a prime to denote the spacelike complement of a region, $\A' = \{p : p \not\in \hat{J}^+(\A)\, \text{and}\, p\not\in \hat{J}^-(\A) \}$.
\end{itemize*}   

\vspace{0.3cm}
\noindent \textbf{Overview of the paper}
\vspace{0.3cm}

In section \ref{sec:geometric}, we begin by establishing $I(\V_1:\V_2\,|\,\U)$ is well defined given the definitions of regions $\V_1,\V_2$ and $\U$, which requires $\V_1$ and $\V_2$ do not overlap, and touch $\U$ only at their boundaries. We also review the maximin formula, and using maximin understand the geometric conditions for $I(\V_1:\V_2\,|\,\U)=O(1/G_N)$. Finally we give the bulk, geometric proof of theorem \ref{thm:main}.  

In section \ref{sec:tasksperspective}, we begin by briefly reviewing holographic quantum tasks. We also discuss in detail how secret messages can be sent in a spacetime context, and the correlations associated with doing so. Finally we turn to an argument for the privacy-duality theorem from this quantum tasks perspective. 

In section \ref{sec:examples} we give a few examples of applications of the privacy-duality theorem to various bulk solutions. This elucidates why the uniqueness condition of the privacy-duality theorem is necessary, and why it is the conditional mutual information and not the mutual information that appears. 

In section \ref{sec:discussion} we conclude with some discussion. In particular, we discuss the relationship of this theorem to the connected wedge theorem, and some suggestions for future directions in which to explore the connections between causal features of bulk geometry and boundary correlation. 

Appendix \ref{sec:saturationlemma} proves a lemma giving conditions on the placement of extremal surfaces for the conditional mutual information to be $O(1/G_N)$, which we make use of in the geometric proof of the privacy-duality theorem.

In appendix \ref{sec:elementarypad} we review an elementary one-time pad (sometimes called "the" one-time pad) useful for sending qubits using private, perfectly correlated classical bits. This serves as useful intuition for our theorem. 

In appendix \ref{sec:whyQCMI?} we show that the privacy-duality theorem is false if the conditional mutual information $I(\V_1:\V_2\,|\,\U)$ is replaced with the mutual information $I(\V_1:\V_2)$. This can be seen by constructing a counterexample in the BTZ black hole geometry. In the quantum information picture we understand why this theorem should be false by finding a method for sending secret messages which maintains $I(\V_1:\V_2)=0$. 

\section{Bulk perspective on the privacy-duality theorem}\label{sec:geometric}

In this section we develop the geometric proof of the privacy-duality theorem. To begin we lay out some preliminary facts about the regions $\V_1,\V_2,\U$, and some background on the Ryu-Takayanagi formula and its generalizations.

\subsection{Geometric preliminaries}\label{sec:geometricprelim}

We recall the maximin formula \cite{wall2014maximin,akers2020quantum}, which is one way of stating the generalization of the Ryu-Takayanagi formula to dynamic spacetimes. 

Consider a spacelike boundary region $\A$. Then the maximin formula states that the von Neumann entropy of the state on $\A$ is given by
\begin{align}
    S(\A) = \max_{\Sigma} \min_{\gamma \in \Sigma} \left(\frac{\text{Area}[\gamma]}{4G_N} + S_b(E[\gamma])\right),
\end{align}
where the maximization is over bulk Cauchy surfaces $\Sigma$ that include region $\A$, the minimization is over codimension 2 surfaces $\gamma$ inside of $\Sigma$ and homologous to $\A$, and $S_b(E[\gamma])$ is the bulk entropy associated with the region $E[\gamma]$ in $\Sigma$ which is enclosed by $\A\cup \gamma$.

The quantity optimized over in the maximin formula is often called the \emph{generalized entropy}, 
\begin{align}
    S_{gen}[\gamma]\equiv \left(\frac{\text{Area}[\gamma]}{4G_N} + S_b(E[\gamma])\right).
\end{align}
We briefly recall the meaning of homologous. We say a surface $\gamma$ is homologous to a boundary region $\A$ if there exists a codimension 1 surface $S$ such that
\begin{align}
    \partial S = \A \cup \gamma.
\end{align}
We will call the surface picked out by the maximin procedure a \emph{quantum extremal surface} and label it by $\gamma_{\mathcal{A}}$. The region $E[\gamma_{\A}]$ picked out by the quantum extremal surface also plays an important role. We will define
\begin{align}
    E_\A \equiv D(E[\gamma_\A]).
\end{align}
The region $E_\A$ is called the entanglement wedge of $\A$. We will assume that the bulk entropy $S_b(E_\A)$ term is smaller than $O(1/G_N)$, although this can be violated in the context of evaporating black holes \cite{penington2020entanglement}. Note also that there are settings where quantum extremal surfaces do not calculate the von Neumann entropy correctly \cite{akers2020leading}. We will assume throughout this paper that we are considering states where the quantum maximin formula holds. 

The \emph{quantum conditional mutual information} (QCMI) is defined by
\begin{align}
    I(\A:\C|\B) \equiv S(\A\B) + S(\B\C) - S(\B) - S(\A\B\C).
\end{align}
We can calculate each term in this expression using the maximin formula. However in general the entropy of a subregion is infinite, with the infinite contributions associated with the boundaries of the various subregions. We will establish that $I(\V_1:\V_2\,|\,\U)$ is well defined and in particular finite. For this to be the case, we need that the regions $\V_1,\V_2,\U$ are all spacelike separated, although coincident boundaries between $\V_1$ and $\U$ and $\V_2$ and $\U$ are allowed. The necessary statements are given in the next lemma. 
\begin{lemma}\label{lemma:v1v2upositions}
Assume there is no curve in the boundary from region $\C$ to region $\R$ which is private with respect to $\U$. Then the following statements are true:
\begin{enumerate*}
    \item $\V_1 \cap \V_2=\varnothing$, that is $\V_1$ and $\V_2$ never overlap or touch. 
    \item $\V_1, \V_2$ never overlap $\U$ except possibly at the spacelike boundary of $\U$.
\end{enumerate*}
\end{lemma}
\begin{proof}
Suppose $\V_1\cap \V_2$ is not empty. Choose a point $p$ in $\V_1\cap \V_2$. Then since $p \in\V_1$ it is in the future of $\C$, so there is a causal curve from $\C$ to $p$. Similarly $p\in \V_2$ so $p$ is in the past of $\R$, so there is a causal curve from $p$ to $\R$. Taking these two segments together, we can construct a causal curve $\gamma$ from $\C$ to $\R$. Since $p \in \U'$, the causal curve $\gamma$ cannot pass into $\U$. Thus $\gamma$ is private. But then we have a private curve from $\C$ to $\R$, which violates our assumption. Thus we have $\V_1 \cap \V_2 = \varnothing$. This establishes point 1.

For point 2, note that $\V_1 \cap \U \subseteq \U'\cap \U = \partial \U$.\footnote{Recall that the $\partial$ symbol is used for the spatial boundary.}
\end{proof}

Next, we would like to know when there is an area term contribution to the QCMI. To do this we recall some arguments in \cite{headrick2014general, wall2014maximin}.
\begin{lemma}\label{lemma:QCMIsaturation}
Assume the quantum focusing conjecture \cite{akers2020quantum}. Then if the area terms in $I(\A:\C\,|\,\B)$ cancel, leaving only a possible bulk entropy contribution then all of the following statements hold:
\begin{enumerate*}
    \item There exists a bulk Cauchy surface $\Sigma$ such that all of $\gamma_\B$, $\gamma_{\A\B}$, $\gamma_{\B\C}$, $\gamma_{\A\B\C}$ are contained in $\Sigma$, and further $S_{gen}[\gamma_{\B}], S_{gen}[\gamma_{\A\B}],S_{gen}[ \gamma_{\B\C}], S_{gen}[\gamma_{\A\B\C}]$ are minimal in that slice. 
    \item $\area (\partial E_{\A\B\C}) = \area({\partial [E_{\A\B} \cup E_{\B\C}]})$
    \item $\area(\partial E_{\B}) = \area(\partial [E_{\A\B} \cap  E_{\B\C}])$
    \item $\partial E_{\B\C}\cap \partial E_{\A\B}=\emptyset$
\end{enumerate*}
\end{lemma}
We give the proof in appendix \ref{sec:saturationlemma}. In the proof of the privacy-duality theorem, we will see that the existence of a private curve implies that either condition 1 or condition 3 is violated. 

To complete the proof of the privacy-duality theorem, we will also need a lemma constraining the placement of the entanglement wedge. 
\begin{lemma}\label{lemma:EWnesting}
Given boundary regions $\A$, $\B$, with $\A \subseteq \B$, we have $E_{\A}\subseteq E_{\B}$. 
\end{lemma}
For a proof see \cite{wall2014maximin,akers2020quantum}. The proof relies on the quantum focusing conjecture. With these preliminaries in hand we move on to prove the privacy-duality theorem in the next section.

\subsection{Geometric proof}

In this section we give the geometric proof of the privacy-duality theorem. We repeat the theorem for convenience. Recall that we say an extremal surface is unique if all other extremal surfaces homologous to the same boundary region have generalized entropy larger than the minimal one at order $O(1/G_N)$

\vspace{0.3cm}
\noindent \textbf{Theorem \ref{thm:main}:}\emph{\textbf{(Privacy-duality)} Consider domains of dependence $\C$, ${\R}$, and ${\U}$ in the boundary of an asymptotically locally AdS spacetime, along with the corresponding bulk entanglement wedges $E_\C$, $E_\R$, and $E_\U$. Assume there is no private curve from $\C$ to $\R$ with respect to $\U$ in the boundary geometry. Define boundary regions
\begin{align}
    \V_1 &= \hat{D}(\hat{J}^+(\C) \cap \U' \cap \partial \Sigma), \\
    \V_2 &= \hat{D}(\hat{J}^-(\R) \cap \U' \cap \partial \Sigma),
\end{align}
where $\partial \Sigma$ is any Cauchy surface for the boundary which includes a Cauchy surface for $\U$. Then if there is a private curve in the bulk from $E_\C$ to $E_\R$ with respect to $E_\U$, and the minimal extremal surface homologous to $\U$ is unique, it follows that $I(\V_1:\V_2\,|\,\U) = O(1/G_N)$.}
\vspace{0.3cm}

\begin{proof}
The four conditions in lemma \ref{lemma:QCMIsaturation} must all hold for the area terms to cancel and leave only bulk entropy terms, so we need to show at least one of those conditions fails. Our strategy will be to show that if we assume condition 1 (that there is a single bulk Cauchy surface containing all the extremal surfaces) then condition 3 fails. 

The argument is as follows. 
\begin{enumerate}
    \item \label{item:Pdefined} By assumption a private curve exists, which we label $\Gamma_P$. We've also assumed a bulk Cauchy surface $\Sigma$ containing all the extremal surfaces exists, which we label $\Sigma$. Define $P \equiv \Sigma\cap \Gamma_P$.
    \item Note that $\C \subseteq \hat{D}(\V_1\cup \U)$. To see this, recall that the boundary domain of dependence $\hat{D}(\V_1\cup \U)$ is the set of all points $p$ such that all boundary causal curves through $p$ intersect $\V_1\cup \U$. But $ \hat{J}^+(\C) \cap \partial \Sigma \subseteq \V_1\cup \U$, so every causal curve through $\C$ reaches $\V_1\cup \U$, so $\C \subseteq \hat{D}(\V_1\cup \U)$. 
    \item \label{item:1inclusion} We will show $P\subseteq E_{\V_1\U}$. First note that $\Gamma_P$ begins inside of $E_\C$, and since $\C \subseteq \hat{D}(\V_1\cup \U)$, entanglement wedge nesting (lemma \ref{lemma:EWnesting}) gives that $\Gamma_P$ begins inside of $E_{\V_1\U}$. But $\Gamma_P$ is a causal curve, and $E_{\V_1\U}$ is the domain of dependence of a codimension 1 surface $\Sigma \cap E_{\V_1 \U}$. Thus $\Gamma_P$ must cross through $\Sigma \cap E_{\V_1 \U}$, so we find $P \subseteq
    \Sigma \cap E_{\V_1 \U}$, which gives $P\subseteq E_{\V_1\U}$.
    \item \label{item:2inclusion} By a similar argument to the above, now using that $\R\subseteq \hat{D}(\V_2 \cup \U)$, we have that $P\subseteq E_{\V_2\U}$.  
    \item Combining points \ref{item:1inclusion} and \ref{item:2inclusion}, we have $P \in E_{\V_1\U}\cap E_{\V_2\U}$. But also by the definition of a private curve, we have that $P \not \in E_\U$. This gives $E_\U \neq E_{\V_1\U}\cap E_{\V_2\U}$.
    \item Finally, note that if minimal extremal surfaces are unique, then $E_\U\neq E_{\V_1\U}\cap E_{\V_2\U}$ implies $\area(\partial E_\U) \neq \area(\partial[E_{\V_1\U}\cap E_{\V_2\U}])$, as needed. 
\end{enumerate}
\end{proof}

\begin{figure}
    \centering
    \begin{tikzpicture}[scale=0.9]
    
    \node[right] at (3,0) {$\V_2$};
    \node[left] at (-3,0) {$\V_1$};
    
    \draw[thick] (0,0) circle (3);
    
    \draw[blue, thick] (2.12, 2.12) to [out=-135,in=135] (2.12, -2.12);
    \draw[blue, thick] (-2.12, 2.12) to [out=-45,in=45] (-2.12, -2.12);
    
    \draw[red,thick] (2.12, 2.12) to [out=-135,in=-45] (-2.12, 2.12);
    \draw[red,thick] (2.12, -2.12) to [out=135,in=45] (-2.12, -2.12);
    \node at (0,1.7) {$\gamma_\U$};
    \node at (0,-1.7) {$\gamma_\U$};
    
    \node at (-1.75,0) {$\gamma_{\V_2\U}$};
    \node at (1.75,0) {$\gamma_{\V_1\U}$};
    
    \node[above] at (0,3) {$\U$};
    \node[below] at (0,-3) {$\U$};
    
    \draw[domain=-45:-135,ultra thick, red] plot ({3*cos(\x)},{(3*sin(\x))});
    \draw[domain=45:135,ultra thick, red] plot ({3*cos(\x)},{(3*sin(\x))});
    
    \draw[domain=-45:45,ultra thick, blue] plot ({3*cos(\x)},{(3*sin(\x))});
    \draw[domain=135:225,ultra thick, blue] plot ({3*cos(\x)},{(3*sin(\x))});
    
    \node at (0,0) {$\otimes$};
    \node[right] at (0,0) {$P$};
    
    \end{tikzpicture}
    \caption{An example set-up for the proof of the privacy-duality theorem. We've taken a simple case, where $\U$ has only two connected components and the overall state is pure. In the case shown there is a bulk Cauchy slice $\Sigma$ which contains both of $\gamma_{\V_1 \U}$ and $\gamma_{\V_2\U}$. In this case, the proof shows that a private curve $\Gamma_P$ crosses $\Sigma$ at a point $P$ which is inside $E_{\V_1\U}$, inside $E_{\V_2\U}$, and outside $E_\U$. It follows that $E_{\U}\neq E_{\V_1\U} \cap E_{\V_2\U}$. Whenever the minimal extremal surface homologous to $\U$ is unique, this implies $I(\V_1:\V_2\,|\,\U)$ is $O(1/G_N)$.}
    \label{fig:slice}
\end{figure}
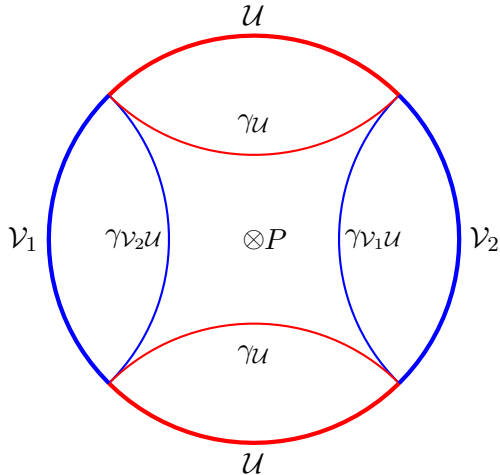

An example surface $\Sigma$ with the positioning of extremal surfaces described in the proof is given in figure \ref{fig:slice}. Note however that the proof applies more generally than the figure. For example it includes higher dimensions, allows $\U$ to have an arbitrary number of connected components, and includes cases where $\V_1\cup\V_2\cup \U$ does not make up the entire boundary.

The proof above relies on the quantum focusing conjecture to establish entanglement wedge nesting (lemma \ref{lemma:EWnesting}). As well, it relies on the maximin formula holding to relate bulk geometric properties to the value of the quantum conditional mutual information. We could also understand the theorem without using the maximin formula however, and in particular without reference to the quantum conditional mutual information or the von Neumann entropy. In particular, viewed as a geometric statement, the proof gives that the existence of a private curve implies one of conditions 1 or 3 in lemma \ref{lemma:QCMIsaturation} fail. The entanglement wedges appearing in those conditions can be understood as geometrically defined objects and need not reference the entropy in any dual theory.

\section{Quantum information perspective on the privacy-duality theorem}\label{sec:tasksperspective}

In this section we argue for the privacy-duality theorem from a quantum information perspective, in particular using the quantum tasks framework of \cite{may2021quantum}. 

\subsection{Localizing and excluding quantum information}

A quantum task is a quantum computation where the spacetime locations of the input and output systems are specified \cite{kent2012quantum}. We will describe the task in terms of two agencies, Alice and Bob. Alice will receive quantum and classical systems as input, process them in some way, then return the outputs to Bob. Bob will prepare the inputs and give them to Alice, then after receiving Alice's outputs will check the outputs are as intended.

To make this notion of a quantum task more precise, we recall some definitions and observations from \cite{hayden2019localizing}, which were discussed in the holographic context in \cite{may2021quantum}. We begin with a definition of a quantum system being localized to a spacetime region. 

\begin{definition}\label{def:localized}
Suppose one party, Alice, holds system $X$, whose purification we label $\ket{\Psi}_{\bar{X}{X}}$. Then we say the subsystem $X$ is \textbf{localized} to a spacetime region $\R$ if a second party, Bob, for whom the state is initially unknown can prepare the $X$ system by acting on $\R$ with some channel $\mathcal{M}_{\R\rightarrow X}$. 
\end{definition}
Conversely, we will say that a system $X$ is \textbf{excluded} from a region $\R$ if Bob cannot learn anything about $X$ by accessing that region.

To make these notions more precise, we will give a quantitative condition for when a system is localized to or excluded from a subregion. To state it, we need to define two quantities. First is the \emph{mutual information}, defined by
\begin{align}
    I(\A:\B) \equiv S(\A) + S(\B) - S(\A\B).
\end{align}
Second, we define the \emph{fidelity},
\begin{align}
    F(\rho,\sigma) \equiv \text{tr} \sqrt{\sqrt{\rho}\,\sigma \sqrt{\rho}}.
\end{align}
The fidelity is equal to one if and only if $\rho=\sigma$. Two states having fidelity near one means they are nearly indistinguishable, in a sense that can be made precise \cite{wilde2013quantum}.

Using these quantities, we can characterize when it is possible to recover a quantum system from a subregion by making use of the following theorem.

\begin{theorem}\label{thm:approximaterecovery}
Consider a quantum channel $\mathcal{N}_{X\rightarrow A}$. Then there is an approximate inverse channel $(\mathcal{N}_{X\rightarrow A})^{-1}$ in the sense that
\begin{align}
    F(\ket{\Psi^+}_{\bar{X}X},\mathcal{I}_{\bar{X}}\otimes (\mathcal{N}_{X\rightarrow A})^{-1}\circ \mathcal{N}_{X\rightarrow A}(\ket{\Psi^+}_{\bar{X}X})) \geq 1-\sqrt{\epsilon}
\end{align}
if and only if
\begin{align}
    I(\bar{X}:A)_{\mathcal{I}\otimes \mathcal{N}(\ket{\Psi^+})} \geq 2S(\bar{X})-\epsilon
\end{align}
where $\ket{\Psi^+}_{\bar{X}X}$ is the maximally entangled state. 
\end{theorem}
See, for example, \cite{schumacher2002approximate} for a proof.\footnote{Note that often this theorem is stated using the coherent information $I(A\rangle B)=S(B)-S(AB)$, but this is simply related to the mutual information, $I(A:B)=I(A\rangle B)+ S(A)$.} Note that the condition in this theorem, that the inverse channel works well on the maximally entangled state, also implies that the inverse channel works well on average \cite{horodecki1999general,nielsen2002simple}. 

To apply this theorem in the context of our definition of localizing a state to a region, we take whatever process has encoded the system $X$ into the region $\R$ to be the channel $\mathcal{N}$. Then theorem \ref{thm:approximaterecovery} informs us that $X$ is localized to that region whenever, taking $\ket{\Psi^+}_{\bar{X}X}$ to be maximally entangled, the mutual information satisfies
\begin{align}
    I(\bar{X}:\R) \geq 2S(\bar{X})-\epsilon.
\end{align}
Conversely, for $\R$ to be excluded we should have
\begin{align}
    I(\bar{X}:\R) \leq \epsilon,
\end{align}
since this implies the complementary system to the degrees of freedom in $\R$ has large mutual information and so recovers $X$, and so $\R$ holds no information about $X$.

Importantly, where a quantum system is localized to or excluded from in the bulk and in the boundary are related, as is captured in the notion of entanglement wedge reconstruction \cite{czech2012gravity,headrick2014causality,wall2014maximin,jafferis2016relative,dong2016reconstruction,cotler2019entanglement}. Informally, a quantum system $X$ is localized to the entanglement wedge $E_\A$ if and only if it is localized to $\A$. More precisely, the recovery from the boundary subregion is approximate, but we can account for this using the condition for approximate recovery given in theorem \ref{thm:approximaterecovery}. The statement of entanglement wedge reconstruction we will use is
\begin{align}\label{eq:EWR}
    I(\bar{X}:E_\A) \geq 2S(\bar{X}) \Longrightarrow I(\bar{X}:\A) \geq 2S(\bar{X}) - \epsilon.
\end{align}
The parameter $\epsilon$ is small in the sense that $\epsilon \rightarrow 0$ when $G_N\rightarrow 0$. 

Considering excluded regions, the situation is similar. To address it consider a subregion $\U$ of the boundary along with its entanglement wedge $E_\U$. Our starting point is that $X$ excluded from $E_\U$ means 
\begin{align}
    I(\bar{X}:E_\U) = 0 .
\end{align}
Introduce an auxiliary AdS space, the purpose of which is to hold the purification of the initial AdS space's state. Then we have
\begin{align}
    I(\bar{X}:E_\U^c) = 2S(\bar{X})
\end{align}
where $E_\U^c$ is the complement of $E_\U$ taken in the full bulk geometry, including that of the auxiliary AdS. Now apply our statement of entanglement wedge reconstruction \ref{eq:EWR} to find
\begin{align}
    I(\bar{X}:\U^c) \geq 2S(\bar{X})-\delta.
\end{align}
But using purity again we have $I(\bar{X}:\U) \leq \delta $, so that $X$ is excluded from $\U$. 

\subsection{The secret message task}

We will consider a scenario where a quantum system $X$ is initially localized to a region $\C$, and should be evolved in such a way that it is later localized to a spacetime region $\R$. Additionally, we require $X$ remain excluded from a third region $\U$. We call this scenario a secret message task, as we detail below. 

\begin{definition}
A secret message task $S^{\times n}_{\epsilon, \delta}$ is defined by
\begin{enumerate*}
    \item An input region $\C$, in which system $X$ is initially localized. 
    \item An output region $\R$, into which Alice should bring system $X$. 
    \item An excluded region $\U$. 
\end{enumerate*}
System $X$ is in a maximally entangled state with a reference system $\bar{X}$, and consists of $n$ qubits. Completing the task successfully requires:
\begin{enumerate*}
    \item \textbf{Secrecy:} $X$ cannot be recovered from region $\U$, so that $I(\bar{X}:\U) \leq \delta$
    \item \textbf{Correctness:} $X$ can be recovered from region $\R$, so that $I(\bar{X}:\R)\geq 2n - \epsilon$.
\end{enumerate*}
\end{definition}
So that Bob can verify that the recovered system $X$ is or is not maximally entangled with $\bar{X}$, we give him the $\bar{X}$ system. Bob will either access $\U$ and check $X$ is not localized there, or access $\R$ and confirm that $X$ is localized there.\footnote{For pedagogical purposes, in the introduction we included a third party, Eve, who may access the region $\U$. Here we have Bob play the eavesdroppers role.}

There are two strategies we will consider for completing this task: a local strategy and a non-local strategy. In the local strategy, $X$ is sent along a causal curve from $\C$ to $\R$ that avoids $\U$. In the non-local strategy, we exploit correlation to conceal information that is sent through $\U$. We discuss each separately below. 

\subsubsection*{Local approach to $S^{\times n}_{\epsilon,\delta}$ task}

If there is a causal curve $\Gamma_P$ from $\C$ to $\R$ which avoids region $\U$, which we denote a \emph{private curve}, then the task is straightforward to complete: simply send $X$ along that curve. More concretely, $X$ should be recorded into some localized degrees of freedom, which then travel along $\Gamma_P$.

If $X$ could be sent noiselessly, then we would have $I(\bar{X}:\R)=2n$. More realistically, some amount of noise will occur in the process of sending $X$ to $\R$, in which case $I(\bar{X}:\R)$ will be smaller. In particular if the channel taking $X$ from $\C$ to $\R$ is denoted $\mathcal{N}$, then it being close to the identity in the sense that
\begin{align}
    \max_{\ket{\psi}_{\bar{X}X}}| \ket{\psi}_{\bar{X}X} - \mathcal{I}_A\otimes \mathcal{N}(\ket{\psi}_{\bar{X}X}) |\leq \gamma
\end{align}
implies $I(\bar{X}:\R) \geq 2n - \gamma^2/4$. We will choose the noise parameter $\epsilon$ in our secret message task to satisfy $\epsilon > \gamma^2/4$, so that we can complete this task simply in the bulk.\footnote{It would also be possible to reduce noise in sending a message along the private curve using error correction techniques, at the cost of increasing the number of qubits sent.}

\subsubsection*{Non-local approach to the $S^{\times n}_{\epsilon,\delta}$ task}

Even without a private curve from $\C$ to $\R$, it is possible to complete the secret message task. Typically, this is done by sharing a secret, classical key consisting of a string of bits between the sender and the receiver. The sender uses the key $k$ to apply a unitary $P_k$ to the message, which she then sends through the region $\U$. The set of unitaries $\{P_k\}$ is chosen so that averaging over them returns the maximally mixed state, concealing the message from an eavesdropper. The receiver then uses their copy of $k$ to apply $P_k^{-1}$, and recover the message. We describe the simplest such protocol in detail in appendix \ref{sec:elementarypad}.

The above protocol for sending secret messages exploits correlation in the key to achieve its goal. We can understand some general requirements on correlations for the secret message task to be completed. To understand this, consider figure \ref{fig:boundaryregions}, where we've drawn the regions $\U, \C, \R$, and
\begin{align}
    \V_1 &= \hat{D}(\hat{J}^+(\C) \cap \U' \cap \partial \Sigma), \\
    \V_2 &= \hat{D}(\hat{J}^-(\R) \cap \U' \cap \partial \Sigma).
\end{align}
Using these definitions of the regions, and theorem \ref{thm:approximaterecovery}, we arrive at the following three statements. 
\begin{itemize}
    \item \textbf{Unitarity:} $I(\bar{X}:\V_1\U)\geq 2n-\epsilon$. Follows because $\V_1 \cup \U$ contains $\hat{J}^+(\C)\cap \partial \Sigma$, into which $X$ must be localized by unitarity, so $X$ is localized to $\V_1\cup \U$.
    \item \textbf{Secrecy:} $I(\bar{X}:\U)\leq \delta$. As part of the definition of the secret message task, we require $X$ be excluded from $\U$, which is just this statement. 
    \item \textbf{Correctness:} $I(\R:\U\V_2) \geq 2n-\epsilon$. This follows because the task being successful requires $X$ be localized to $\R$, but $\U\cup \V_2$ contains $\hat{J}^-(\R)\cap \partial \Sigma$, and by unitarity $X$ is localized to its past light cone. 
\end{itemize}
Using these three statements we can prove an interesting bound on the correlations among the subsystems $\V_1,\U,\V_2$, which we give in the following lemma. 

\begin{lemma}\label{lemma:QCMIbound}
Completing the secret message task $S^{\times n}_{\delta, \epsilon}$ when there are no private curves from $\C$ to $\R$ requires
\begin{align}
    I(\V_1:\V_2\,|\,\U) \geq 2n - \epsilon'
\end{align}
where $\epsilon' = 2\epsilon+\delta$.
\end{lemma}
\begin{proof}
It will be convenient to introduce a system $Y$ which purifies $\bar{X}\V_1\U\V_2$. In terms of the purification system, unitarity and correctness become
\begin{align}
    I(\bar{X}:\V_2Y) &\leq \epsilon, \\
    I(\bar{X}:\V_1Y) &\leq \epsilon.
\end{align}
Now we use these entropic statements and basic entropy inequalities to bound the conditional mutual information. 
\begin{align}
    I(\V_1:\V_2|\U) &= S(\V_1\U)+S(\V_2\U) - S(\U) - S(\V_1\V_2\U) \,\,\,\,\,\,\,\,\,\,\,\,\,\,\,\,\,\,\,\,\,\,\,\,\,\,\,\,\,\,\,\,\,\,\,\,\,\,\,\,\,\,\,\,\,\,\,\, \text{(definition)} \nonumber \\
    &= S(\bar{X}\V_2Y) +S(\V_1\bar{X}Y) - S(\U) - S(\bar{X}Y) \,\,\,\,\,\,\,\,\,\,\,\,\,\,\,\,\,\,\,\,\,\,\,\,\,\,\,\,\,\,\,\,\,\,\,\,\,\,\,\,\,\,\, \text{(purity)} \nonumber \\
    &\geq S(\V_2\bar{X}Y)+S(\V_1\bar{X}Y)-S(\U) - S(\bar{X}) - S(Y) \,\,\,\,\,\,\,\,\,\,\,\,\,\,\,\,\,\,\,\,\,\,\,\,\,\,\, \text{(subadditivity)} \nonumber \\
    &\geq S(\V_2\bar{X}Y)+S(\V_1\bar{X}Y)-S(\U \bar{X}) - S(Y) - \delta \,\,\,\,\,\,\,\,\,\,\,\,\,\,\,\,\,\,\,\,\,\,\,\,\,\,\,\,\,\,\,\,\,\text{(secrecy)} \nonumber \\
    &\geq S(\V_2\bar{X}Y)+S(\V_1\bar{X}Y)-S(\V_1\V_2Y) - S(Y) - \delta \,\,\,\,\,\,\,\,\,\,\,\,\,\,\,\,\,\,\,\,\,\,\,\,\,\, \text{(purity)} \nonumber \\
    &\geq S(\V_2\bar{X}Y)+S(\V_1\bar{X}Y)-S(\V_1Y) - S(\V_2Y) - \delta \,\,\,\,\,\,\,\,\,\,\,\,\,\,\,\,\,\,\,\,\,\,\,\,\,\, \text{(S.S.A.)} \nonumber \\
    &\geq 2S(\bar{X}) - \delta - 2\epsilon \,\,\,\,\,\,\,\,\,\,\,\,\,\,\,\,\,\,\,\,\,\,\,\,\,\,\,\,\,\,\,\,\,\,\,\,\,\,\,\,\,\,\,\,\,\,\,\,\,\,\,\,\,\,\,\,\,\,\,\,\,\,\,\,\,\,\,\,\,\,\,\,\,\,\,\,\,\,\,\,\,\,\,\,\,\,\,\,\,\,\,\,\,\,\,\,\,\,\,\,\,\,\,\,\,\,\,\, \text{(unitarity and correctness)} \nonumber 
\end{align}
Dividing by $2$ and noting that $X\bar{X}$ is a maximally entangled state on $n$ qubits, we obtain the claimed bound.
\end{proof}

This lemma, applied to the boundary CFT, will be the key result in arguing for the privacy-duality theorem in the next section. 
\subsection{Quantum tasks argument for privacy-duality}

In this section we give the quantum tasks argument for the privacy-duality theorem. We repeat the theorem here for convenience. 
\vspace{0.3cm}

\noindent \textbf{Theorem \ref{thm:main}:}\emph{\textbf{(Privacy-duality)} Consider domains of dependence $\C$, ${\R}$, and ${\U}$ in the boundary of an asymptotically locally AdS spacetime, along with the corresponding bulk entanglement wedges $E_\C$, $E_\R$, and $E_\U$. Assume there is no private curve from $\C$ to $\R$ with respect to $\U$ in the boundary geometry. Define boundary regions
\begin{align}
    \V_1 &= \hat{D}(\hat{J}^+(\C) \cap \U' \cap \partial \Sigma), \\
    \V_2 &= \hat{D}(\hat{J}^-(\R) \cap \U' \cap \partial \Sigma),
\end{align}
where $\partial \Sigma$ is any Cauchy surface for the boundary which includes a Cauchy surface for $\U$. Then if there is a private curve in the bulk from $E_\C$ to $E_\R$ with respect to $E_\U$, and the minimal extremal surface homologous to $\U$ is unique, it follows that $I(\V_1:\V_2\,|\,\U) = O(1/G_N)$.}
\vspace{0.3cm}

\begin{argument}
By assumption, there is a private curve which passes through the bulk from a point $c$ in $E_{\C}$ to a point $r$ in $E_{\R}$. The curve is private because it avoids the region $E_\U$. We define a secrecy task in the bulk with $c$ as input location, $r$ as output location, and $E_\U$ as excluded region. Call this task ${S}^{\times n}_{\delta, \epsilon}$. 

Because a private curve exists in the bulk, we can use the local strategy to complete the ${S}^{\times n}_{\delta, \epsilon}$ task. The parameter $\epsilon$ should be chosen to reflect any unavoidable noise present in carrying this out, and $\delta$ chosen to reflect any information about the message which becomes available in $E_\U$. Note that a successful protocol requires we send $n$ qubits through the bulk. 

Next consider the boundary picture. We consider a secrecy task $\hat{S}^{\times n}_{\delta, \epsilon}$ in the boundary which has $\C$ as input region, $\R$ as output region, and $\U$ as excluded region. We see that any process which completes the bulk task ${S}^{\times n}_{\delta, \epsilon}$ maps in the boundary to a process which completes the boundary task $\hat{S}^{\times n}_{\delta, \epsilon}$, since $c\in \C$, $r\in \R$ and excluding $X$ from $E_\U$ in the bulk means excluding it from $\U$ in the boundary. Since ${S}^{\times n}_{\delta, \epsilon}$ can be completed then, so can $\hat{S}^{\times n}_{\delta, \epsilon}$. 

Now we apply lemma \ref{lemma:QCMIbound}, to learn that
\begin{align}\label{eq:lowerboundonperturbed}
    I(\V_1:\V_2\,|\,\U)_{\tilde{\psi}} \geq 2n - \epsilon'
\end{align}
Here $\tilde{\psi}$ is the state of the CFT on $\partial \Sigma$ when the $n$ qubits are sent through the bulk. We'll label the unperturbed state, without the $n$ qubits sent in the bulk, by $\psi$. We would like to bound the QCMI in the unperturbed state. To relate this to the bound above on the perturbed state, consider each term in the perturbed QCMI,
\begin{align}
    I(\V_1:\V_2\,|\,\U)_{\tilde{\psi}} &= S(\V_1\U)_{\tilde{\psi}} + S(\V_2\U)_{\tilde{\psi}} - S(\U)_{\tilde{\psi}} - S(\V_1\V_2\U)_{\tilde{\psi}}
\end{align}
By assumption the private curve does not enter $E_{\U}$, in the original state $\psi$. So long as the minimal extremal surface homologous to $\U$ is unique, and we restrict to $n<O(1/G_N)$, this remains true in the state $\tilde{\psi}$.\footnote{We study an example where the minimal extremal surface is degenerate in section \ref{sec:examples}} Given this, the additional qubits do not enter $E_\U$ in the perturbed state, so do not contribute their entropy to $S(\U)_{\tilde{\psi}}$. On the other hand, by the unitarity property the message is localized to $\V_1\cup \U$, so the entropy of the additional $n$ qubits do contribute to $S(\V_1\U)_{\tilde{\psi}}$. They also then contribute to the strictly larger regions entropy, $S(\V_1\V_2\U)_{\tilde{\psi}}$. Similarly, by the correctness property they contribute to $S(\V_2\U)_{\tilde{\psi}}$. This means the entropy of the additional qubits, $n$, contributes positively in $I_{\tilde{\psi}}$ twice and negatively once, so that
\begin{align}
    I(\V_1:\V_2\,|\,\U)_{\tilde{\psi}} = I(\V_1:\V_2\,|\,\U)_{\psi} + n.
\end{align}
Combining this with inequality \ref{eq:lowerboundonperturbed}, we have
\begin{align}
    I(\V_1:\V_2\,|\,\U)_\psi \geq n-\epsilon'.
\end{align}
Since $n$ can be any order less than $O(1/G_N)$, we conclude that the QCMI in the unperturbed state $\psi$ is $O(1/G_N)$. 
\end{argument}

We have referred to the above as an argument because of the statement that it is possible to complete the local strategy in the bulk. In particular, it is necessary to record $X$ into some local degrees of freedom, then route those appropriately through spacetime. If the private curve is a geodesic this is a fairly innocuous assumption. For more general curves, it might be necessary to, for example, arrange an appropriate system of mirrors to reflect photons along a path that approximates $\Gamma_P$. Given a description of the bulk effective field theory it is not immediately clear this sort of construction is possible (in the same way, given the Standard Model, it is not clear that one can construct mirrors, etc.). We assume however that the bulk theory allows an appropriate construction. Outside of this assumption however, this argument is a proof. In particular there are no loopholes like the one discussed for the connected wedge theorem \cite{may2020holographic}.

In appendix \ref{sec:whyQCMI?} we discuss in detail why it should be the conditional mutual information, rather than the mutual information, that appears in the privacy-duality theorem. Indeed, we find counterexamples in the setting where the conditional mutual information is replaced by the mutual information.

\section{Privacy-duality theorem examples} \label{sec:examples}

We give a few examples in this section which highlight various features of the privacy-duality theorem. 

\subsection{Case where mutual information is small}\label{sec:notmutualinfo}

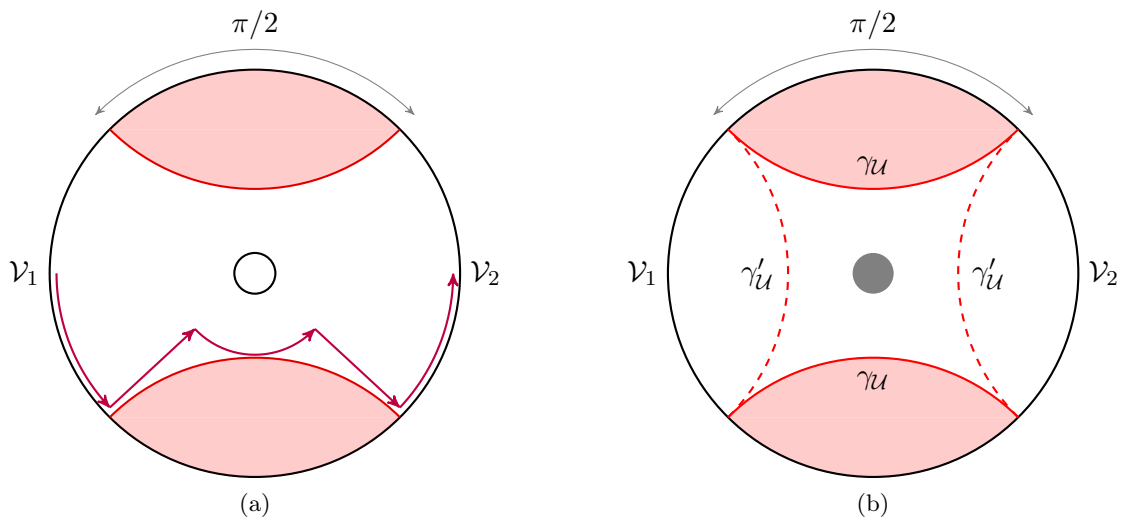
\begin{figure}
    \centering
    \subfloat[\label{fig:BTZexample}]{
    \begin{tikzpicture}[scale=0.9]
    
    \draw[thick] (0,0) circle (3);
    \draw[thick] (0,0) circle (0.3);
    
    \draw[red, thick] (-2.12, 2.12) to [out=-45,in=-135] (2.12, 2.12);
    \draw[red, thick] (-2.12, -2.12) to [out=45,in=135] (2.12, -2.12);
    
    \draw[opacity=0.2,fill=red,domain=45:135]  (-2.12, 2.12) to [out=-45,in=-135] (2.12, 2.12) plot ({3*cos(\x)},{3*sin(\x)});
    \draw[opacity=0.2,fill=red,domain=-45:-135]  (-2.12, -2.12) to [out=45,in=135] (2.12, -2.12) plot ({3*cos(\x)},{3*sin(\x)});
    
    \node[right] at (3,0) {$\V_2$};
    \node[left] at (-3,0) {$\V_1$};
    
    \draw[<->,gray,domain=135:45] plot ({3.3*cos(\x)},{3.3*sin(\x)});
    \node[above] at (0,3.3) {$\pi/2$};
    
    \draw[thick, purple,domain=-180:-137,->] plot ({2.9*cos(\x)},{2.9*sin(\x)});
    
    \draw[thick, purple,->] ({2.9*cos(-137)},{2.9*sin(-137)}) -- ({1.2*cos(-137)},{1.2*sin(-137)});
    
    \draw[thick, purple,domain=-137:-43,->] plot ({1.2*cos(\x)},{1.2*sin(\x)});
    
    \draw[thick, purple,->] ({1.2*cos(-43)},{1.2*sin(-43)}) -- ({2.9*cos(-43)},{2.9*sin(-43)});
    
    \draw[thick, purple,domain=-43:0,->] plot ({2.9*cos(\x)},{2.9*sin(\x)});
    
    \end{tikzpicture}
    }
    \hfill
    \subfloat[\label{fig:dustexample}]{
    \begin{tikzpicture}[scale=0.9]
    
    \draw[opacity=0.2,fill=red,domain=45:135]  (-2.12, 2.12) to [out=-45,in=-135] (2.12, 2.12) plot ({3*cos(\x)},{3*sin(\x)});
    \draw[opacity=0.2,fill=red,domain=-45:-135]  (-2.12, -2.12) to [out=45,in=135] (2.12, -2.12) plot ({3*cos(\x)},{3*sin(\x)});
    
    \draw[thick] (0,0) circle (3);
    \fill[gray] (0,0) circle (0.3);
    
    \draw[red, thick] (-2.12, 2.12) to [out=-45,in=-135] (2.12, 2.12);
    \draw[red, thick] (-2.12, -2.12) to [out=45,in=135] (2.12, -2.12);
    
    \node[right] at (3,0) {$\V_2$};
    \node[left] at (-3,0) {$\V_1$};
    
    \draw[<->,gray,domain=135:45] plot ({3.3*cos(\x)},{3.3*sin(\x)});
    \node[above] at (0,3.3) {$\pi/2$};
    
    \draw[red,dashed,thick] (-2.12,-2.12) to [out=45,in=-45] (-2.12,2.12);
    \draw[red,dashed,thick] (2.12,-2.12) to [out=135,in=-135] (2.12,2.12);
    
    \node[left] at (-1.3,0) {$\gamma_{\U}'$};
    \node[right] at (1.3,0) {$\gamma_{\U}'$};
    
    \node[above] at (0,1.3) {$\gamma_{\U}$};
    \node[below] at (0,-1.3) {$\gamma_{\U}$};
    
    \end{tikzpicture}
    }
    \caption{(a) Top down view of a causal curve constructed in the BTZ geometry from point c at $\phi=0$, $t=-3\pi/4$ to the boundary at $\phi=\pi$. For a small enough black hole, the curve reaches $\phi=\pi$ before the point $r$, which is at time $3\pi/4$. The curve is constructed to avoid $E_\U$ (light red). See appendix \ref{sec:whyQCMI?} for details. (b) An example with degenerate minimal extremal surfaces, constructed by starting with AdS$_{2+1}$. A small amount of matter with $O(1)$ entropy has been added to the bulk. Surfaces $\gamma_\U$ (red) and $\gamma_\U'$ (dashed red) have the same area. }
    \label{fig:examples}
\end{figure}

In the case where the the state of $\V_1\V_2\U$ is pure the conditional mutual information and mutual information are equal, $I(\V_1:\V_2\,|\,\U)=I(\V_1:\V_2)$. We might plausibly expect then that the mutual information should be large whenever there is a private curve. While this is immediately true for pure states, this idea fails for mixed states. 

An example can be constructed in the BTZ black hole geometry, as we illustrate in figure \ref{fig:BTZexample}. Region $\U$ consists of two intervals of size $\pi/2$ and centered at $\phi=\pi/2,3\pi/2$ on the $t=0$ time slice. Regions $\C$ and $\R$ are chosen to be points, and located at $\phi=0,\pi$ and $t=-3\pi/4$, $t=3\pi/4$ respectively. In appendix \ref{sec:whyQCMI?}, we find that for small enough black holes there is a causal, private curve from $c$ to $r$. In this geometry $I(\V_1:\V_2)=0$, showing the mutual information need not be large. One can check that the conditional mutual information is large however, as required by the privacy-duality theorem. 

Based on this example, we should anticipate that there are one-time pad protocols that maintain small mutual information even while sending a private message. Indeed this is the case, as we show in appendix \ref{sec:whyQCMI?}. 

\subsection{Case with degenerate areas}

In the privacy-duality theorem we require that the minimal extremal surface for region $\U$ be unique. Here we give an example to show that this is in fact necessary. In particular we find an example where a private curve exists, and $I(\V_1:\V_2\,|\,\U)$ is $O(1)$ by allowing degenerate minima. 

Consider the arrangement of regions shown in figure \ref{fig:dustexample}, which shows a constant time slice of global AdS$_{2+1}$. Region $\U$ consists of two intervals of size $\pi/2$ and centered at $\phi=\pi/2,3\pi/2$. The input and output regions are chosen to be points, $\C=c$ and $\R=r$, and located at $\phi=0,\pi$ and $t=-3\pi/4$, $t=3\pi/4$ respectively. A causal curve can easily travel from $c$ to $r$: A light ray traveling radially inward from $c$ reaches $\phi=\pi$ at time $\pi/4$, well before point $r$. 

Now add a small amount of matter with $O(1)$ entropy near $r=0$. Then the light ray starting at $c$ is delayed a small amount, but will still reach $\phi=\pi$ before $r$, so there is a causal curve from $c$ to $r$ in the bulk. Additionally adding this matter ensures the curve is private: there are two minimal surfaces of equal area enclosing $\U$, which we label $\gamma_\U$ and $\gamma_\U'$. To minimize the area plus entropy, the quantum maximin formula selects the disconnected configuration for $E_\U$. This ensures a curve through the bulk does not enter the entanglement wedge of $\U$. 

Label the area of one segment of $\gamma_\U$ or $\gamma_\U'$ by $A$, and consider the conditional mutual information:
\begin{align}\label{eq:QCMIisSb}
    I(\V_1:\V_2\,|\,\U) &= S(\V_1\U) + S(\V_1\U) - S(\U) - S(\V_1\V_2\U), \nonumber \\
    &= [A + S_b] + [A+S_b] - 2A - S_b, \nonumber \\
    &= S_b,
\end{align}
so that $I(\V_1:\V_2\,|\,\U)=O(1)$, which establishes that it was indeed necessary to assume the minimal surfaces were non-degenerate. 

It is interesting to understand where the gravity and quantum information arguments break down in this case. In the gravity proof, without assuming uniqueness of minimal surfaces, we can still conclude that, whenever there is a bulk Cauchy slice containing $\gamma_{\U}, \gamma_{\V_1\U}, \gamma_{\V_2\U}, \gamma_{\V_1\V_2\U}$, then
\begin{align}
    E_\U \neq E_{\V_1 \U} \cap E_{\V_2 \U}. 
\end{align}
Indeed we see this holds in our example. From here however we cannot conclude that the two boundaries of these regions have different areas. Meanwhile in the quantum information argument, we find that inserting $n$ qubits in the bulk may move the location of the entanglement wedge $E_\U$, even for $n<O(1/G_N)$. Thus while the original state $\psi$ may have a private curve, the perturbed state may not, and we cannot complete the argument. In the example of figure \ref{fig:dustexample}, we could send $n<S_b$ qubits before $E_\U$ moves and removes all private curves from the geometry. Thus in that case we still learn that
\begin{align}
    I(\V_1:\V_2\,|\,\U)_{{\psi}} \geq S_b
\end{align}
Which as we saw in equation \ref{eq:QCMIisSb} holds and is actually saturated.

\subsection{Counterexample to converse of the theorem}

We can ask if the converse statement to the privacy-duality theorem holds. That is, if $I(\V_1:\V_2\,|\,\U)=O(1/G_N)$, does that imply the existence of a private curve? We will show the answer is no. 

Take region $\U$ to consist of two intervals of size $\pi/4$ and centered at $\phi=\pi/2,3\pi/2$. The input and output regions are chosen to be points $\C=c$ and $\R=r$, and located at $\phi=0,\pi$ and $t=-\pi/2$, $t=\pi/2$ respectively. Then there is just barely a causal curve from $c$ to $r$: the light ray sent radially inward from $c$ ends on $r$. Now, add a small amount of matter to the bulk. This will delay the light ray from $c$ to $r$, so that now there is no private curve. At the same time, the conditional mutual information is $O(1/G_N)$. 

We have found that there exist choices of $\C,\R, \U$ such that the converse to the privacy-duality theorem does not hold. A more interesting statement however is the following. Given regions $\V_1, \V_2, \U$, and assuming $I(\V_1:\V_2\,|\,\U)=O(1/G_N)$, does there always exist a choice of regions $\C, \R$ such that $\V_1=\hat{D}(J^+(\C)\cap \U'\cap \partial \Sigma)$, $\V_2=\hat{D}(J^-(\R)\cap \U'\cap \partial \Sigma)$ and there is a private curve from $\C$ to $\R$? We have not yet constructed a counterexample to this statement. We do not expect it holds however, on the grounds that the analogous statement for the similar connected-wedge theorem does not hold. 

\section{Discussion}\label{sec:discussion}

The connection between entanglement and geometry in AdS/CFT has advanced our understanding of how gravitational physics can be recorded into a quantum mechanical system. For instance, Einsteins equations can be derived from entanglement physics \cite{lashkari2014gravitational,swingle2014universality,faulkner2017nonlinear,lewkowycz2018holographic}, and bulk energy conditions follow from entropy inequalities \cite{lashkari2015inviolable,lashkari2016gravitational}. While the Ryu-Takayanagi formula relates boundary entanglement to bulk extremal surfaces, the connected wedge theorem and privacy-duality theorems add a complementary perspective on the entanglement-geometry connection in relating boundary entanglement to bulk light cones. 

The privacy-duality theorem strengthens the connection between bulk light cones and boundary entanglement in a number of ways. For instance, the quantum tasks arguments for both theorems are not entirely rigorous. Having two instances where thinking about quantum tasks leads to true statements (they are verified by geometric proofs) supports the validity of the quantum tasks approach. Additionally, we may be able to find connections between the two theorems, or to find commonalities between them that suggest novel directions. We explore some of these ideas in the comments below. 

Finally, we note that from a geometric perspective we might have expected a connection between bulk extremal surfaces and light cones. On the basis of the Ryu-Takayanagi formula, we could then expect a relationship between bulk light cones and boundary entanglement. What is surprising however is that there is a separate, direct link between bulk light cones and boundary entanglement coming from the quantum information perspective. 

\subsection{Relationship with the connected wedge theorem}

Recall that the connected wedge theorem \cite{may2020holographic} says that the existence of a scattering region implies boundary regions $\V_1,\V_2$ have $I(\V_1:\V_2)=O(1/G_N)$. The connected wedge theorem can be understood as following from entanglement requirements for performing quantum teleportation\footnote{Earlier references \cite{may2019quantum,may2020holographic} instead used entanglement requirements for performing non-local quantum computations, but this is mainly for technical reasons (stronger bounds are known in that case). Similar arguments can be made instead using teleportation. See the introduction of \cite{may2021thesis} for a sketch of this. }, while the privacy-duality theorem is related to the quantum one-time pad. Teleportation and the one-time pad are closely related constructions. As one example, teleportation sends messages secretly, since the two classical measurement outcomes which are transmitted are uncorrelated with the message qubit. Given this relationship in the quantum information aspects of the two theorems, it is plausible there is also a geometric relationship between the two theorems.

One observation is that in cases where the state on $\V_1\V_2\U$ is pure, the conditional mutual information is equal to the mutual information, $I(\V_1:\V_2\,|\,\U)=I(\V_1:\V_2)$. It may be simplest to establish a geometric relationship between the two theorems in this setting. 

\subsection{Causal conditions for other entropy inequalities}

The connected wedge theorem gives a causal perspective on the transition in the mutual information from $O(1)$ to $O(1/G_N)$, and the privacy-duality theorem gives a causal perspective on a similar transition in the conditional mutual information. The mutual information is associated with the subadditivity property, $I(A:B)\geq 0$, while the conditional mutual information is associated with strong subadditivity $I(A:C|B)\geq 0$. 

Given this parallel between the connected wedge and privacy-duality theorems, it is natural to speculate that other quantities associated with entropy inequalities will similarly be related to a causal condition. Perhaps the next simplest quantities are those associated with the Araki-Lieb inequality or the monogamy of mutual information \cite{headrick2014general}. Other entropy inequalities are discussed in \cite{bao2015holographic}. 

\subsection{Implications for quantum information theory}

In the case of the connected wedge theorem, the connection between a quantum task and AdS/CFT had implications in both directions. The necessity of entanglement for a non-local computation task implied the connected wedge theorem, a novel result in holography. In the other direction, the possibility of local computations happening in the bulk was argued to imply certain non-local computations can be performed with entanglement linear in the input size, less entanglement than the best known constructions in many cases \cite{may2020holographic}. 

In the case of the privacy-duality theorem, we can also ask what the implications are for quantum information theory. Consider that when a bulk private curve exists, we can send a private message consisting of $n < O(1/G_N)$ qubits. Then in the boundary we know the conditional mutual information is at most order $O(1/G_N)$, so can conclude no more than linear conditional mutual information is involved in sending a private message. Of course, existing constructions (like those given in the appendices) achieve this, so we have not learned anything new. Nonetheless this argument still seems valuable, for two reasons. First, it is a very similar argument to the one given for the connected wedge theorem, and reaches a conclusion that is known independently to be correct. In this way it indirectly supports the argument for linear entanglement in non-local computation. Secondly, it is an interesting thought experiment to understand how much about quantum information theory is implied by the geometry of AdS spacetimes. 

Another implication for quantum information from AdS/CFT occurring in this context arises in the BTZ geometries discussed in appendix \ref{sec:whyQCMI?}. In that context one can find examples where there is a private curve but the mutual information $I(\V_1:\V_2)$ is small. We were previously not aware it was possible to send secret messages with small mutual information, but found this to be implied by these geometries. Motivated by this gravitational construction, we found the quantum information protocol given in appendix \ref{sec:whyQCMI?}. While this construction is simple and could have been found in other ways, it is interesting that the route to this quantum information protocol started with an observation about AdS geometries. We are hopeful that there is more to be learned about quantum information by pursuing further holographic quantum tasks. 

\vspace{0.3cm}
\noindent \textbf{Acknowledgements}
\vspace{0.3cm}

I thank Kfir Dolev, Sam Cree and Mark Van Raamsdonk for helpful discussions. Aidan Chatwin-Davies provided valuable feedback on this manuscript. I am supported by a C-GSM award given by the National Science and Engineering Research Council of Canada. 

\appendix 

\section{Proof of lemma \ref{lemma:QCMIsaturation}}\label{sec:saturationlemma}

We repeat the lemma here for reference, and outline the proof below. The proof relies on some results from \cite{akers2020quantum}, which we cite as we go. Note that our proof consists of assembling some comments from \cite{headrick2014general, wall2014maximin, akers2020quantum}, but to our knowledge is not stated explicitly elsewhere.
\vspace{0.2cm}

\noindent \emph{\textbf{Lemma \ref{lemma:QCMIsaturation}:}
Assume the quantum focusing conjecture \cite{akers2020quantum}. Then if the area terms in $I(\A:\C\,|\,\B)$ cancel, leaving only a possible bulk entropy contribution then all of the following statements hold:
\begin{enumerate*}
    \item There exists a bulk Cauchy surface $\Sigma$ such that all of $\gamma_\B$, $\gamma_{\A\B}$, $\gamma_{\B\C}$, $\gamma_{\A\B\C}$ are contained $\Sigma$, and further $S_{gen}[\gamma_{\B}], S_{gen}[\gamma_{\A\B}],S_{gen}[ \gamma_{\B\C}], S_{gen}[\gamma_{\A\B\C}]$ are minimal in that slice. 
    \item $\area (\partial E_{\A\B\C}) = \area({\partial [E_{\A\B} \cup E_{\B\C}]})$
    \item $\area(\partial E_{\B}) = \area(\partial [E_{\A\B} \cap  E_{\B\C}])$
    \item The part of $\gamma_{\A\B} \cap \gamma_{\B\C}$ which has $E_{\A\B}$ on one side and $E_{\B\C}$ on the other is empty.
\end{enumerate*}}

\begin{proof}
First, suppose that condition 1 holds. Then define a region $H_{\B}$ by $D(E_{\A\B}\cap E_{\B\C}\cap \Sigma)$ and a second surface by $H_{\A\B\C} = D([E_{\A\B}\cup E_{\B\C}]\cap \Sigma)$. Define the boundaries $\tilde{\gamma}_{\B} = \partial H_\B$ and $\tilde{\gamma}_{\A\B\C}=\partial H_{\A\B\C}$. One may then check that 
\begin{align}\label{eq:intersections}
    \area(\partial H_{ABC}) + \area(\partial H_B) = \area(\partial E_{AB}) + \area( \partial E_{BC}) - 2\,\area(\partial E_{\A\B}\cap E_{\B\C}).
\end{align}
Next, recall that the conditional mutual information is given by
\begin{align}
    I(\A:\C|\B) &= S_{gen}[E_{\A\B}] + S_{gen}[E_{\B\C}] - S_{gen}[E_{\B}] - S_{gen}[E_{\A\B\C}]
\end{align}
We will add and subtract the term $S_{gen}[H_{\A\B\C}]+ S_{gen}[H_\B]$ from this expression, 
\begin{align}
    I(\A:\C|\B) &= \left(S_{gen}[E_{\A\B}] + S_{gen}[E_{\B\C}] -S_{gen}[H_{\A\B\C}]-S_{gen}[H_\B]\right) \nonumber \\
    & + \left(S_{gen}[H_{\A\B\C}] - S_{gen}[E_{\A\B\C}] \right) \nonumber \\
    & + \left(S_{gen}[H_\B] -S_{gen}[E_{\B}] \right)
\end{align}
Lets consider each line here separately. In the first line, we can use \ref{eq:intersections} to see that this is given by $2 \,\area(\partial E_{\B\C}\cap \partial E_{\A\B})$ plus a bulk entropy term. The second and third lines are separately positive, because $S_{gen}[E_{\A\B\C}]$ and $S_{gen}[E_\B]$ are minimal for their respective boundary regions. Given this, we can conclude that $\partial E_{\B\C}\cap \partial E_{\A\B} = \emptyset$ for the overall expression to lack an area term. This gives condition $4$. Further, the area terms in the second and third lines must separately be zero. This gives conditions $2$ and $3$. 

Finally, we consider the case where condition 1 fails, so that there is no single Cauchy slice in which all of the surfaces are minimal. Then an argument from \cite{wall2014maximin,akers2020quantum} establishes that there will be an area term. In brief the argument is as follows: 
\begin{itemize*}
    \item Show there is always a surface in which both $S_{gen}[E_\B]$ and $S_{gen}[E_{\A\B\C}]$ are minimal (See theorem 5 of \cite{akers2020quantum}).
    \item Use $E_{\A\B}$ and $E_{\B\C}$ to construct regions $E_{\A\B}'$ and $E_{\B\C}'$ which lie in $\Sigma$, have less generalized entropy, and are homologous to $\A\B$ and $\B\C$ respectively. In fact these will have less generalized entropy at $O(1/G_N)$ (See section 3.2 of \cite{akers2020quantum}, note that this step assumes the quantum focusing conjecture).
    \item Use the above argument to show that $S_{gen}[E_{\A\B}] + S_{gen}[E_{\B\C}] - S_{gen}[E_{\B}] - S_{gen}[E_{\A\B\C}]$ is non-negative, at least considering the $O(1/G_N)$ terms.\footnote{In fact this quantity is positive, even considering bulk entropy terms, though we didn't show this above. See section 4.2 of \cite{akers2020quantum}.}
\end{itemize*}
This establishes that the conditional mutual information will be at least as large as $(S_{gen}[E_{AB}]-S_{gen}[E_{\A\B}'])+(S_{gen}[E_{\B\C}]-S_{gen}[E_{\B\C}'])=O(1/G_N)$ whenever condition 1 doesn't hold. 
\end{proof}

\section{The elementary one-time pad}\label{sec:elementarypad}

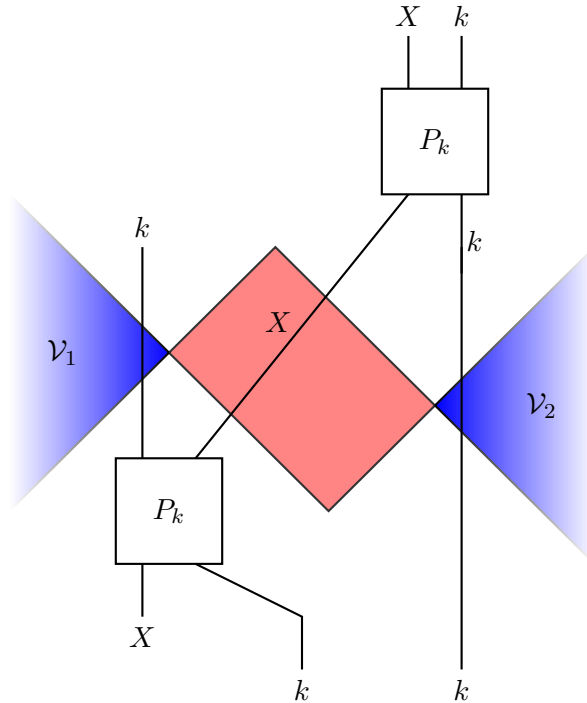
\begin{figure}
    \centering
    
    \begin{tikzpicture}[scale=0.7]
    
    \draw[fill=red!60!,opacity=0.8,thick] (1,4) -- (3,6) -- (6,3) -- (4,1) -- (1,4);
    \filldraw[thick, color = black, fill = blue, path fading = east] (6,3) -- (9,6) -- (9,0) -- (6,3);
    \filldraw[thick, color = black, fill = blue, path fading = west] (1,4) -- (-2,7) -- (-2,1) -- (1,4);
    
    \draw[black,thick] (0,0) rectangle (2,2);
    \node at (1,1) {$P_k$};
    
    \draw[black,thick] (5,7) rectangle (7,9);
    \node at (6,8) {$P_k$};
    
    \draw[thick] (0.5,-1) -- (0.5,0);
    
    \draw[black,thick] (1.5,0) -- (3.5,-1) -- (3.5,-2); 
    
    \draw[black,thick] (6.5,-2) -- (6.5,6);
    
    
    \draw[thick] (6.5,5.5) -- (6.5,7);
    \node[below right] at (6.4,6.5) {$k$};
    
    \draw[black, thick] (0.5,2) -- (0.5,6); 
    \node[above] at (0.5,6) {$k$};
    
    \draw[black,thick] (5.5,9) -- (5.5,10);
    \node[above] at (5.5,10) {$X$};
    
    \draw[black,thick] (6.5,9) -- (6.5,10);
    \node[above] at (6.5,10) {$k$};
    
    \draw[thick] (1.5,2) -- (5.5,7);
    \node[above left] at (3.5,4.2) {$X$};
    
    \node[below] at (0.5,-1) {$X$};
    
    \node[below] at (3.5,-2) {$k$};
    \node[below] at (6.5,-2) {$k$};
    
    \node at (-1,4) {$\V_1$};
    \node at (8,3) {$\V_2$};
    
    \end{tikzpicture}

    \caption{The elementary one-time pad, used in a spacetime context. A pair of correlated, classical bit strings $k$ is prepared at an early time. One copy is sent through $\V_2$ to near $r$. Another copy is sent near $c$, then onward through $\V_1$. Near $c$ the string $k$ is used to encode the message system $X$. The encoded $X$ is sent through $\U$, then decoded using the copy of $k$ near $r$. This procedure leads to large $I(\V_1:\V_2)$ and large $I(\V_1:\V_2\,|\,\U)$.}
    \label{fig:onetimepadinpsacetime}
\end{figure}

Suppose Alice holds a single qubit on a system labelled $X$. We'll assume this is in a maximally entangled state with a reference qubit $\bar{X}$. Alice will attempt to send $X$ to Bob so that Bob holds maximal entanglement with $\bar{X}$. Alice can send qubits to Bob, but only over a public quantum channel. This means that Eve, an eavesdropper, can choose to either allow Alice's message to go through untouched, or to intercept her transmission and receive the qubits intended for Bob. In this section we study a simple protocol which sends the qubit $X$ securely over a public channel using 2 shared, perfectly correlated classical bits as a resource. The generalization to sending $n$ qubits is obvious. This is a special but instructive case of the more general one-time pad constructions discussed in the main text. 

Label the two private, classical, random bits shared by Alice and Bob by $k=(k_1,k_2)$. Then they can use the following procedure to send the qubit $X$ securely. 
\begin{protocol}\textbf{Elementary quantum one-time pad:}
\begin{enumerate*}
    \item Alice applies $P_k=X_{k_1}Z_{k_2}$ to $X$. 
    \item Alice inputs $X$ to the public quantum channel.
    \item Bob applies $P_k=X_{k_{1}}Z_{k_{2}}$ to $X$.
\end{enumerate*}
\end{protocol}
We need to verify that this protocol is both secure (Eve can't learn anything when she intercepts) and correct (Bob recovers $\ket{\psi}$ when Eve does not intercept). To understand both of these, consider the $\bar{X}XAB$ joint state after Alice has applied the encoding procedure,
\begin{align}
    \rho_{\bar{X}XAB} = \frac{1}{4} \sum_k (\mathcal{I}\otimes P_k) \ketbra{\Psi^+}{\Psi^+}_{\bar{X}X}(\mathcal{I}\otimes P_k)\otimes \ketbra{k}{k}_A\otimes \ketbra{k}{k}_{B}.
\end{align}
If Eve intercepts she gains the $X$ system, while $A$, $B$ are always held by Alice and Bob respectively. A straightforward calculation reveals $\rho_{\bar{X}X}=\mathcal{I}/4$, so that $I(\bar{X}:X)_{\rho}=0$ and Eve has learned nothing about the secret message. This establishes that the protocol is secure. To see it is also correct, notice that when Bob receives $X$ he can measure $B$ to learn $k$, then apply the $P_k$ operator again so that he holds the $X$ subsystem of $\ket{\Psi^+}_{\bar{X}X}$, and so holds maximal entanglement with $\bar{X}$. 

We can use the one-time pad protocol described above to complete the secret message task. An embedding of the protocol into spacetime is shown in figure \ref{fig:onetimepadinpsacetime}. Notice that the $A$ system passes through the $\V_1$ region, and $B$ passes through $\V_2$. Thus using this protocol we would have large mutual information $I(\V_1:\V_2)$, in addition to a large conditional mutual information. More generally however the large mutual information is not necessary, as we discuss in the next appendix. 

\section{Why the conditional mutual information?}\label{sec:whyQCMI?}

\begin{figure}
    \centering
    
    \begin{tikzpicture}[scale=0.7]
    
    \draw[fill=red!60!,opacity=0.8,thick] (1,4) -- (3,6) -- (6,3) -- (4,1) -- (1,4);
    \filldraw[thick, color = black, fill = blue, path fading = east] (6,3) -- (9,6) -- (9,0) -- (6,3);
    \filldraw[thick, color = black, fill = blue, path fading = west] (1,4) -- (-2,7) -- (-2,1) -- (1,4);
    
    \draw[black,thick] (0,0) rectangle (2,2);
    \node at (1,1) {$P_k$};
    
    \draw[black,thick] (5,7) rectangle (7,9);
    \node at (6,8) {$P_k$};
    
    \draw[thick] (0.5,-1) -- (0.5,0);
    
    \draw[black,thick] (1.5,0) -- (3.5,-1) -- (3.5,-2); 
    \draw[black,thick] (6.5,1) -- (6.5,4.5);
    
    \draw[black,thick] (5.5,0) rectangle (7,1);
    \node at (6.25,0.5) {$\oplus $};
    
    \draw[black,thick] (6.5,-2) -- (6.5,0);
    
    \node at (7,3) {$l$};
    \node at (4.25,3) {$k\oplus l$};
    
    \draw[black,thick] (6,1) -- (5,2.5) -- (5,3.25) -- (6,4.5);
    
    \draw[thick] (5.5,4.5) rectangle (7,5.5);
    \node at (6.25,5) {$\oplus $};
    
    \draw[thick] (6.5,5.5) -- (6.5,6) -- (7.5,7) -- (7.5,10);
    \node[below left] at (6,6.5) {$k$};
    \node[above] at (7.5,10) {$l$};
    
    \draw[black, thick] (0.5,2) -- (0.5,6); 
    \node[above] at (0.5,6) {$k$};
    
    \draw[black,thick] (5.5,9) -- (5.5,10);
    \node[above] at (5.5,10) {$X$};
    
    \draw[black,thick] (6.5,9) -- (6.5,10);
    
    \node[above] at (6.5,10) {$k$};
    
    \draw[thick] (1.5,2) -- (5.5,7);
    \node[above left] at (3.5,4.2) {$X$};
    
    \node[below] at (0.5,-1) {$X$};
    
    \node[below] at (3.5,-2) {$k$};
    \node[below] at (6.5,-2) {$k$};
    
    \draw[black,thick] (6,5.5) -- (6,7);
    
    \draw[black,thick] (6,-2) -- (6,0);
    \node[below] at (6,-2) {$l$};
    
    \node at (-1,4) {$\V_1$};
    \node at (8,3) {$\V_2$};
    
    \end{tikzpicture}
    \caption{At an early time the bit strings $l$ and $k$ are prepared. Their sum mod 2, $l\oplus k$, is sent directly through $\U$ to near $r$, while $l$ is sent through $\V_2$ to near $r$. Meanwhile, $k$ is sent near $c$ then onwards through $\V_1$. Near $c$ the string $k$ is used to encode the message, which is then sent through $\U$. Meanwhile, near $r$, $l$ and $l\oplus k$ are used to calculate $k$, which is then used to decode the message. This procedure leads to $I(\V_1:\V_2)=0$ and $I(\V_1:\V_2\,|\,\U)$ large. }
    \label{fig:trickyonetimepad}
\end{figure}

In the simplest realization of the one-time pad (see appendix \ref{sec:elementarypad}), the mutual information, as well as the conditional mutual information, is as large as the message size. One might ask then why the conditional mutual information rather than the mutual information appears in the privacy-duality theorem. In this appendix we explain that 1) from a bulk perspective, the existence of a private curve does not imply large mutual information, and provide an explicit example and 2) in the boundary it is possible to send secret messages while keeping $I(\V_1:\V_2)$ small. 

Let's begin with the bulk perspective. We will construct a geometry and choice of regions $\C,\R,\U$ that has a private curve but $I(\V_1:\V_2)=0$. Consider the BTZ black hole geometry, with metric
\begin{align}
    ds^2 = - (r^2-r_+^2) dt^2 + \left( \frac{1}{r^2-r_+^2}\right)dr^2 + r^2 d\theta^2.
\end{align}
We have set the AdS length to $1$. Take the region $\U$ to be two intervals on the $t=0$ surface of angular size $\pi/2$. Locate them antipodally, centered on $\theta=\pi/2,3\pi/2$. Note that the entanglement wedge $E_\U$ will consist of two separated regions, and not enclose the black hole. Take the input and output regions $\C$, $\R$ to be points $c$, $r$ located at $(\theta,t)$ positions
\begin{align}
    c = (0,-3\pi/4+\epsilon), \\
    r = (\pi,3\pi/4-\epsilon).
\end{align}
For $\epsilon$ small and positive, there is no private curve in the boundary. Next construct causal (non-geodesic) curves as follows. 
\begin{enumerate*}
    \item Begin at $c$ and travel along the boundary until reaching $\theta=\pi/4$
    \item Travel radially inward until reaching $r=r_\U$, the deepest $r$ value reached by the minimal surface for one segment of $\U$. 
    \item Then travel at constant $r$ until reaching $\theta=3\pi/4$.
    \item Travel outward radially to the boundary.
    \item Travel along the boundary until reaching $\theta=\pi$.
\end{enumerate*}
See figure \ref{fig:BTZexample}. By construction this avoids entering the entanglement wedge of $\U$. The deepest radius reached by the minimal surface attached to one segment of $\U$ is given by
\begin{align}
    r_\U = r_+ \cosh\left( \frac{r_+\pi}{2} \right).
\end{align}
Using this we can find the time taken for the above causal curve,
\begin{align}
    \Delta t = \frac{\pi}{2} + \frac{r_+\pi}{2} + \frac{\pi}{2} \cosh\left(\frac{r_+\pi}{2}\right).
\end{align}
This should be less than the time difference between $c$ and $r$, which is $3\pi/2$. Thus this curve reaches $r$ and is private if $\Delta t \leq 3\pi/2$, which simplies to
\begin{align}\label{eq:rplusrequirement}
    r_+ + \cosh\left( r_+ \pi/2\right) \leq 2.
\end{align}
For a small enough black hole this is satisfied. 

Now consider the regions $\V_1$ and $\V_2$, which are intervals of size $\pi/2$ centered at $\theta=0$ and $\theta=\pi$. The entanglement wedge for $\V_1\cup \V_2$ will be disconnected so that $\gamma_{\V_1\V_2}=\gamma_{\V_1}\cup \gamma_{\V_2}$. This gives $I(\V_1:\V_2)=0$, so that we have a private curve and small mutual information, as claimed could occur. 

In contrast, consider what occurs with the conditional mutual information. One may check that for $r_+$ satisfying equation \ref{eq:rplusrequirement}, the entanglement wedge of an interval of size $3\pi/2$ will include the black hole. In particular $E_{\V_1\U}$, $E_{\U\V_2}$ and $E_{\V_1\V_2\U}$ include the black hole. Further, the black hole is not included in $E_{\U}$ as noted above. This leads to
\begin{align}
    I(\V_1:\V_2\,|\,\U) =  \frac{\text{Area}[\gamma_{BH}]}{4G_N},
\end{align}
so the conditional mutual information is $O(1/G_N)$, as follows from the privacy-duality theorem. 

We can also understand why $I(\V_1:\V_2)$ need not be large from the boundary perspective, by giving a method of sending secret messages that keeps $I(\V_1:\V_2)$ small. Consider the arrangement of regions shown in figure \ref{fig:trickyonetimepad}. At an early time, prepare the state
\begin{align}
    \rho_{WYZ} = \sum_{k,l} \ketbra{k}{k}_W\otimes \ketbra{l\oplus k}{l\oplus k}_Y\otimes \ketbra{l}{l}_Z.
\end{align}
To see how to make use of this resource state, consider figure \ref{fig:trickyonetimepad}. System $W$ is sent through $\C$ then $\V_1$, $Z$ though $\V_2$ and then to $r$, and $Y$ is sent directly through region $\U$ to $r$. In region $\C$ the encoding procedure described for the elementary one time pad is performed using the bit $k\oplus l$ as key. The encoded qubits are then sent through $\U$ as well. Then the state on $\V_1\V_2\U$ is
\begin{align}
    \rho_{\V_1\V_2\U \bar{X}} = \sum_{k,l} \ketbra{k}{k}_{\V_1} \otimes [\ketbra{l\oplus k}{l\oplus k}_Y\otimes (\mathcal{I}\otimes P_{k})\ketbra{\Psi^+}{\Psi^+}_{X\bar{X}}(\mathcal{I}\otimes P_{k})] \otimes \ketbra{l}{l}_{\V_2}
\end{align}
where $YX$ together make up the $\U$ system. Notice that $k,l$ are independent and random, so $I(\V_1:\V_2)=0$. Nonetheless this state maintains secrecy and correctness. To see secrecy, notice that $k\oplus l$ is independent of $k$, so although $k\oplus l$ is sent through $\U$ it provides no assistance to Eve in recovering $X$. To see correctness, notice that to decode the message, the receiver can use the $\V_2$ system to learn $l$, then add this to $k\oplus l$ to recover $k$, then uses $k$ to undo $P_k$. 

While the above protocol avoids having large mutual information, a straightforward calculation reveals the conditional mutual information is $2S(\bar{X})$. We can also understand this intuitively by noting that
\begin{align}
    I(\V_1:\V_2\,|\, \U) = I(\V_1:\V_2\U) - I(\V_1:\U).
\end{align}
Thus $I(\V_1:\V_2\,|\, \U)$ is the amount of additional information we learn about $\V_1$ by gaining $\V_2$, if we already know $\U$. In our setting, once we know $l\oplus k$ (which is inside of $\U$), we can learn a lot about $k$ (which is in $\V_1$) by gaining $l$ (which is in $\V_2$). 

\bibliographystyle{unsrt}
\bibliography{biblio.bib}

\end{document}